%% file: paper.tex
\documentclass{lmcs}
\pdfoutput=1

\usepackage{lastpage}
\lmcsdoi{18}{2}{17}
\lmcsheading{}{\pageref{LastPage}}{}{}%
{Apr.~30,~2021}{Jun.~14,~2022}{}

\pdfoutput=1

\usepackage[utf8]{inputenc}
\usepackage[english]{babel}
\usepackage[OT1]{fontenc}
\usepackage{hyperref}
\usepackage{xspace}
\usepackage{amsfonts,amsmath,amssymb,stmaryrd}
\usepackage{inferance}
\usepackage{cleveref}
\usepackage{scalerel}

\definecolor[named]{ACMPurple}{cmyk}{0.55,1,0,0.15}
\definecolor[named]{ACMDarkBlue}{cmyk}{1,0.58,0,0.21}

\usepackage[countsame]{coqtheorem}
\setBaseUrl{{http://www.ps.uni-saarland.de/extras/fol-trakh-ext/website/Undecidability.TRAKHTENBROT.}}

\newcoqtheorem{definition}{Definition}
\newcoqtheorem{lemma}{Lemma}
\newcoqtheorem{theorem}{Theorem}
\newcoqtheorem{corollary}{Corollary}
\newcoqtheorem{myfact}{Fact}
\newcoqtheorem{observation}{Observation}

\input{macros}

\begin{document}

\title[Trakhtenbrot's Theorem in Coq]{Trakhtenbrot's Theorem in Coq: \texorpdfstring{\\}{}
Finite Model Theory through the Constructive Lens\rsuper*}
\titlecomment{{\lsuper*}extended version of~\cite{ijcarversion}}

\author{Dominik Kirst\rsuper{a}}
\address{Saarland University, Saarland Informatics Campus, Saarbrücken, Germany}
\email{kirst@cs.uni-saarland.de}{}{}
\author{Dominique Larchey-Wendling\rsuper{b}}
\address{Université de Lorraine, CNRS,  LORIA, Vand{\oe}uvre-l\`es-Nancy, France}
\email{dominique.larchey-wendling@loria.fr}

\begin{abstract}
We study finite first-order satisfiability (FSAT)
in the constructive setting of dependent type theory.
Employing synthetic accounts of enumerability and decidability,
we give a full classification of FSAT depending on
the first-order signature of non-logical symbols.
On the one hand,
our development focuses on Trakhtenbrot's theorem, stating
that FSAT is undecidable as soon as the signature contains an at least binary
relation symbol. Our proof proceeds by a many-one reduction chain
starting from the Post correspondence problem.
On the other hand, we establish the decidability of FSAT for monadic
first-order logic, i.e.\ where the signature only
contains at most unary function and relation symbols, as well as the enumerability
of FSAT for arbitrary
enumerable signatures.
To showcase an application of Trakhtenbrot's theorem, we continue our reduction chain with a many-one reduction from FSAT to separation logic.
All our results are mechanised in the framework of a growing Coq library of synthetic undecidability proofs.
\end{abstract}

\maketitle

\section{Introduction}

In the wake of the seminal discoveries concerning
the undecidability of first-order logic by Turing and Church in the 1930s,
a broad line of work has been pursued to characterise the border between decidable
and undecidable fragments of the original decision problem.
These fragments can be grouped either by syntactic restrictions
controlling the allowed 
function and relation symbols
or the quantifier prefix, or by semantic restrictions
on the admitted models (see~\cite{borger1997classical} for
a comprehensive description).

Concerning signature restrictions, already predating the undecidability results, Löwenheim had shown in 1915 that monadic first-order logic, admitting only signatures with at most unary symbols, is decidable~\cite{Lowenheim1915}.
Therefore, the successive negative results usually presuppose non-trivial signatures containing an at least binary symbol.

Turning to semantic restrictions, Trakhtenbrot proved in 1950 that, if only admitting finite models, the satisfiability problem over non-trivial signatures is still undecidable~\cite{trakhtenbrot50}.
Moreover, the situation is somewhat dual to the unrestricted case, since finite satisfiability (FSAT) is still enumerable while, in the unrestricted case, validity is enumerable.
As a consequence, finite validity cannot be characterised by a complete finitary deduction system and, resting on finite model theory, various natural problems in database theory and separation logic are undecidable.
The latter will be subject of a case study in Section~\ref{sec:seplog}.

Conventionally, Trakhtenbrot's theorem is proved by (many-one) reduction from the halting problem for Turing machines (see e.g.~\cite{borger1997classical,Libkin:2010:EFM:1965351}).
An encoding of a given Turing machine $M$ can be given as a formula $\phi_M$ such that the models of $\phi_M$ correspond to the runs of $M$.
Specifically, the finite models of $\phi_M$ correspond to terminating runs of $M$ and so a decision procedure for FSAT of $\phi_M$ would be enough to decide whether $M$ terminates or not.

Although this proof strategy is in principle explainable on paper, already the formal definition of Turing machines, not to mention their encoding in first-order logic, is not ideal for mechanisation in a proof assistant.
So for our Coq mechanisation of Trakhtenbrot's theorem, we follow a different strategy by starting from the Post correspondence problem (PCP), a simple matching problem on strings.
Similar to the conventional proof, we proceed by encoding every instance $R$ of PCP as a formula $\phi_R$ such that $R$ admits a solution iff $\phi_R$ has a finite model.
Employing the framework of synthetic undecidability~\cite{ForsterCPP,library_coqpl},
the computability of $\phi_R$ from $R$ is guaranteed since all functions definable in constructive type theory are computable without reference to a concrete model of computation.

Both the conventional proof relying on Turing machines and our elaboration starting from PCP actually produce formulas in a custom signature well-suited for the encoding of the seed decision problems.
The sharper version of Trakhtenbrot's theorem, stating that a signature with at least one binary relation (or one binary function and one unary relation) is enough to turn FSAT undecidable, is in fact left as an exercise in e.g.~Libkin's book~\cite{Libkin:2010:EFM:1965351}.
However, at least in a constructive setting, this generalisation is non-trivial and led us to mechanising a chain of signature transformations eliminating and compressing function and relation symbols step by step.

The constructive and type-theoretic setting introduces subtleties that remain hidden from view in a classical approach to (finite) model theory.
Among these subtleties, quotients are critical for signature reductions but not generally constructively available.
With suitable notions of finiteness (here defined as listability) and discreteness (decidable equality), however, we are able to build finite and decidable quotients (\Cref{thm:fin_quotient}) sufficient for our purposes.
Moreover, the usual set-theoretic constructions in model theory can be simulated in type theory just to some extend, with the prominent lack of a (computable) power set.
Fortunately, it is possible to use the notion of weak power set (\Cref{appendix:lem:wposet}) for these constructions.
Relatedly, finiteness \emph{does not entail} computability in the constructive setting, so the Tarski semantics has to be refined.
In a critically useful result, we establish that finite satisfiability is not impacted by the further requirement of the discreteness of the model (\Cref{thm:quotient_FSAT}).

Complementing the undecidability result, we further formalise that FSAT is enumerable for enumerable signatures and decidable for monadic signatures.
Again, both of these standard results come with their subtleties when explored in a constructive approach to finite model theory.

\smallskip

In summary, the main contributions of this paper are the following:
\begin{itemize}
	\item
	we provide an axiom-free Coq mechanisation comprising a full classification of finite satisfiability with regards to the signatures allowed;\footnote{Downloadable from \url{http://www.ps.uni-saarland.de/extras/fol-trakh-ext/} and systematically hyperlinked with the definitions and theorems in this PDF.}
	\item
	we present a streamlined proof strategy for Trakhtenbrot's theorem well-suited for mechanisation and simple to explain informally, basing on PCP;\@
	\item
	we give a constructive account of signature transformations and the treatment of interpreted equality typically neglected in a classical development;
       \item
        compared to the conference version of this paper~\cite{ijcarversion}, we contribute a refined analysis of the
        conditions allowing for decidable FSAT in the case of monadic signatures,
        introducing the notion of discernability of symbols (e.g.~\Cref{thm:full_monadic_fol});
	\item
	additionally to~\cite{ijcarversion} also, we mechanise a many-one reduction from FSAT
        to the satisfiability problem of separation logic following~\cite{10.1007/3-540-45294-X_10} (e.g.~\Cref{thm:SL});
        \item
        finally, we point out that some of the involved proofs that were just sketched in~\cite{ijcarversion}
        have been expanded (e.g.~\Cref{thm:quotient_FSAT}), and reworked for better readability (e.g.~\Cref{coq:FIN_DISCR_DEC_nSAT_FIN_DEC_2SAT} and~\Cref{lemma:remove_funcs}).
\end{itemize}

\noindent
The paper is structured as follows.
We first describe the type-theoretical framework for undecidability proofs and the representation of first-order logic in Section~\ref{sec:prelims}.
We then outline our variant of Trakhtenbrot's theorem for a custom signature in Section~\ref{sec:trakh_prelim}.
This is followed in Section~\ref{sec:finmod} by a development of enough constructive finite model theory to reach the
stronger form of Trakhtenbrot's theorem where the signature is only assumed to contain one at least binary symbol.
In Section~\ref{sec:decidability} we switch to decidability results for FSAT over monadic signatures,
first assuming that symbols enjoy decidable equality, then maximally strengthening the decidability results
to the case of (the weaker notion of) decidable Boolean discernability. In Section~\ref{sec:trakh_full}, we conclude
with the precise decidability/undecidability classification of FSAT\@.
Section~\ref{sec:seplog} comprises the case study on the undecidability of separation logic and we end with a brief discussion of the Coq development and future work in Section~\ref{sec:discussion}.

\section{First-Order Satisfiability in Constructive Type Theory}%
\label{sec:prelims}

In order to make this paper accessible to readers unfamiliar with constructive type theory, we outline the required features of Coq's underlying type theory, the synthetic treatment of computability available in constructive mathematics, some properties of finite types, as well as our representation of first-order logic.

\subsection{Basics of Constructive Type Theory}

We work in the framework of a constructive type theory such as the one implemented in Coq, providing a predicative hierarchy of type universes $\Type$ above a single impredicative universe $\Prop$ of propositions.
On type level, we have the unit type $\Unit$ with a single element $\unit:\Unit$,
the void type $\Void$, function spaces $X\to Y$, products $X\times Y$, sums $X+Y$, dependent products $\forall x:X.\,F\,x$, and dependent sums $\SigType{x:X}{F\,x}$.
On propositional level, these types are denoted using the usual logical notation ($\top$, $\bot$, $\to$, $\land$, $\lor$, $\forall$, and $\exists$).

We employ the basic inductive types
of Booleans ($\Bool\bnfdef \btrue\mid\bfalse$), of
Peano natural numbers ($n:\Nat\bnfdef 0\mid 1{+}n$),
the option type ($\Opt\,X\bnfdef \some{x}\mid\none$),
and lists ($l:\List\,X\bnfdef\cnil\mid x\ccons l$).
We write $\clen l$ for the \emph{length} of a list, $l\capp m$ for the \emph{concatenation}
of $l$ and $m$, $x\inl l$ for \emph{membership},
and simply $f\map [x_1;\ldots;x_n]\cdef [f\,x_1;\ldots;f\,x_n]$ for the
\emph{map} function.
We denote by $X^n$ the type of vectors of length $n:\Nat$
and by $\Fin n$ the finite types understood as indices
$\{0,\ldots,n-1\}$.
The definitions/notations for lists are shared with vectors $\vec v:X^n$.
Moreover, when $i:\Fin n$ and $x:X$, we denote by $\vec v_i$ the $i$-th component of $\vec v$
and by $\subst{\vec v\,}xi$ the vector $\vec v$ with $i$-th component updated to value $x$.

\subsection{Synthetic (Un-)decidability}
\label{sec:undec}

We review the main ingredients of our synthetic approach to
decidability and undecidability~\cite{forster2018verification,ForsterCPP,forster2019certified,library_coqpl,Larchey-WendlingForster:2019:H10_in_Coq,spies2020undecidability}, based on the computability of all
functions definable in constructive type theory%
\footnote{A result shown and applied for many variants of constructive type theory
and which Coq designers are committed to maintain as Coq evolves.}
or other constructive foundations of mathematics.
We first introduce standard notions of computability theory without referring to a formal model of computation, e.g.\ Turing machines.

\begin{definition}%
\label{def:decidable}
	A \emph{problem} or predicate $p:X\to \Prop$ is
	\begin{itemize}
		\item
		\setCoqFilename{decidable}
		\emph{\coqlink[decidable_bool_eq]{decidable}} if there is $f:X \to\Bool$ with
		$\forall x.\,p\,x\toot f\,x=\btrue$.
		\item
		\setCoqFilename{enumerable}
		\emph{\coqlink[opt_enum_t]{enumerable}} if there is $f:\Nat\to\Opt\,X$ with $\forall x.\,p\,x\toot \exists n.\, f\,n=\some x$.
	\end{itemize}
	These notions generalise to predicates of higher arity.
	Moreover, a type $X$ is
	\begin{itemize}
		\item
		\setCoqFilename{enumerable}
		\emph{\coqlink[type_enum_t]{enumerable}} if there is $f:\Nat\to\Opt\,X$ with $\forall x. \exists n.\, f\,n=\some x$.
		\item
		\setCoqFilename{decidable}
		\emph{\coqlink[discrete]{discrete}} if equality on $X$ (i.e.\ $\lambda xy:X.\,x=y$) is decidable.
		\item
		a \emph{data type} if it is both enumerable and discrete.
	\end{itemize}
\end{definition}

\noindent
\setCoqFilename{decidable}
Using the expressiveness of dependent types, we equivalently tend to establish the decidability
of a predicate $p:X\to\Prop$ by giving a function \coqlink[decidable]{$\forall x:X.\,p\,x +\neg p\,x$}.
Note that it is common to mechanise decidability results in this synthetic sense (e.g.~\cite{braibant2010efficient,maksimovic2015hocore,schafer2015completeness}).
Next, decidability and enumerability transport along reductions:

\setBaseUrl{{http://www.ps.uni-saarland.de/extras/fol-trakh-ext/website/Undecidability.}}
\setCoqFilename{Synthetic.Definitions}
\begin{definition}[][reduces]
	A problem $p:X\to\Prop$ \emph{(many-one) reduces} to $q:Y\to\Prop$, written $p\red q$,
        if there is a function $f: X\to Y$ such that $p\,x\toot q\,(f\,x)$ for all $x:X$.\footnote{Or equivalently, the \setCoqFilename{Synthetic.SyntheticDefinitions}%
        \coqlink[reduction_dependent]{dependently typed characterisation} $\forall x:X.\,\SigType{y:Y}{p\,x\toot q\,y}$.}
\end{definition}

\setBaseUrl{{http://www.ps.uni-saarland.de/extras/fol-trakh-ext/website/Undecidability.TRAKHTENBROT.}}
\begin{myfact}%
	\label{fact:reduction}
	Assume $p:X\to\Prop$, $q:Y\to\Prop$ and $p\red q$:
        \setCoqFilename{decidable}%
        \coqlink[ireduction_decidable]{(1)}~if $q$ is decidable, then so is $p$ and
        \setCoqFilename{enumerable}%
	\coqlink[ireduction_opt_enum_t]{(2)}~if $X$ and $Y$ are data types and $q$ is enumerable,
                then so is $p$.
\end{myfact}

\begin{proof}
	These are Fact 2.11 and Lemma 2.12 in~\cite{ForsterCPP}.
\end{proof}

Item~(1) implies that we can justify the undecidability of a target problem by reduction from a seed problem known to be undecidable, such as the halting problem for Turing machines.
This is in fact the closest rendering of undecidability available in a synthetic setting, since the underlying type theory is consistent with the assumption that every problem is decidable\rlap.\footnote{As witnessed by classical set-theoretic models satisfying $\forall p:\Prop.\, p+\neg p$ ({cf.}~\cite{werner_sets_1997}).}\enspace%
Nevertheless, we believe that in the intended effective interpretation for synthetic computability, a typical seed problem is indeed undecidable and so are the problems reached
by verified reductions.
More specifically, since the usual seed problems are not co-enumerable, (2) implies that the reached problems are not co-enumerable either.

\smallskip

Given its simple inductive characterisation involving only basic types of lists and Booleans,
the (binary) Post correspondence problem (\BPCP) is a well-suited seed problem for compact encoding into first-order logic.

\setBaseUrl{{http://www.ps.uni-saarland.de/extras/fol-trakh-ext/website/Undecidability.TRAKHTENBROT.}}
\setCoqFilename{bpcp}
\begin{definition}
	Given a list $R: \List(\List\,\Bool\times \List\,\Bool)$ of pairs $s/t$ of Boolean strings\rlap,\footnote{Notice that
        the list $R$ is viewed as a (finite) set of pairs $s/t\inl R$ (hence ignoring the order or duplicates),
        while $s$ and $t$, which are also lists, are viewed a strings (hence repetitions and ordering matter for $s$ and
        $t$).}
        we define \emph{derivability} of a pair $s/t$ from $R$ (denoted by \coqlink[pcp_hand]{$R\deriv s/t$})
        and \emph{solvability} (denoted by \coqlink[BPCP_problem]{$\BPCP\,R$}):
	\[\infer1{s/t \inl R}{R\deriv s/t}                               \quad\qquad
	\infer2{s/t \inl R}{R\deriv u/v}{R\deriv (s\capp u)/(t\capp v)}  \quad\qquad
	\infer1{R\deriv s/s}{\BPCP\,R}
        \]
\end{definition}

\begin{myfact}%
  \label{fact:BPCP_undec}
        Given a list $R:\List(\List\,\Bool\times \List\,\Bool)$, the derivability predicate $\lambda s\,t.R\deriv s/t$
        is \coqlink[bpcp_hand_dec]{decidable}. However, the halting problem for Turing machines reduces to \BPCP.
\end{myfact}

\setCoqFilename{red_utils}
\begin{proof}
        We give of proof of the decidability of $R\deriv s/t$ by induction on $\clen s+\clen t$.
        We also provide a \coqlink[BPCP_BPCP_problem_eq]{trivial proof} of the equivalence of two definitions of \BPCP.
        See~\cite{forster2018verification,forster2019certified} for details on the reduction from the halting problem to \BPCP.
\end{proof}

It might at first appear surprising that derivability $\lambda s\,t.R\deriv s/t$ is decidable while \BPCP\ is reducible from the halting
problem (and hence undecidable). This simply illustrates that undecidability is rooted in the unbounded existential quantifier in the equivalence
$\BPCP\,R\toot \exists s.\,  R\deriv s/s$.

\smallskip

In summary, the approach to undecidability used in this and other papers~\cite{ForsterCPP,forster2018verification,forster2019certified,Larchey-WendlingForster:2019:H10_in_Coq,spies2020undecidability,dudenhefner2020undecidability,kirst:2001:synthetic} contributing to the Coq Library of Undecidability Proofs~\cite{library_coqpl} is to verify (synthetic) many-one reductions from a problem known to be undecidable, rooted by the halting problem for Turing machines.
In our case, we start from \BPCP as well-suited seed for finite first-order satisfiability, backed by the reduction from Turing machine halting to \BPCP\ verified in~\cite{forster2018verification}.
Therefore, our mechanised undecidability results as reported in \Cref{sec:trakh_full,sec:seplog} are statements of the form $\BPCP\red P$ for problems $P$ expressing finite first-order satisfiability and satisfiability in separation logic, respectively.
Given the constructive setting, these reductions expressed in Coq's type theory may then be interpreted as computable functions, transporting the undecidability (specifically, non-co-enumerability) of \BPCP\ to said problems, along the lines of Fact~\ref{fact:reduction}.
This improves on the ubiquitous pen-and-paper practice to sketch an algorithm and leave its computability implicit in that, by formally defining the algorithm as a Coq term, its computability is guaranteed by its very construction.

The even more explicit alternative would be to resort to a concrete model of computation, e.g.\
by implementing high-level reduction functions operating on abstract data structures as Turing machines operating over textual or binary
encodings of those data-structures.
Then the notion of undecidability could be formally bootstrapped by showing that there is no Turing machine deciding the
(textually encoded) halting problem, and other problems could be shown undecidable by verifying reductions computable by a Turing machine.
However, this approach would require low-level coding in such a model, which in principle can be supported by tools to a certain extent~\cite{forster_et_al:LIPIcs:2019:11072}, but still introduces enormous overhead unnecessary to cope with in a setting providing its own implicit notion of computability.

\subsection{Constructive Finiteness}

We present four tools for manipulating finite types: the
 \emph{finite pigeon hole principle} (PHP) here established without assuming discreteness,
the \emph{well-foundedness} of strict orders over finite types, \emph{quotients} over strongly decidable
equivalences that map onto $\Fin n = \{0,\ldots,n-1\}$, and the \emph{weak powerset} finitely enumerating every weakly
decidable predicate over a finite type. But first, let us fix a definition of finiteness.

\setBaseUrl{{http://www.ps.uni-saarland.de/extras/fol-trakh-ext/website/Undecidability.Shared.Libs.DLW.Utils.}}
\setCoqFilename{fin_base}
\begin{definition}
	A \emph{type $X$ is \coqlink[finite_t]{finite}} if there is a list $l_X$ s.t. $\forall x:X.\,x\inl l_X$, and
	a \emph{predicate $p:X\to\Prop$ is \coqlink[fin_t]{finite}} if there is a list $l_p $ s.t. $\forall x.\, p\,x\toot x\inl l_p $.
\end{definition}

Note that in constructive settings there are various alternative characterisations of finiteness%
\footnote{And these alternative characterisations are not necessarily constructively equivalent.}
(bijection with $\Fin n$ for some $n$; negated infinitude for some definition of infiniteness; etc)
and we opted for the above since it is easy to work with while transparently capturing the expected meaning.
One can distinguish \emph{strong} finiteness in $\Type$ (i.e.\ $\SigType{l_X:\List\,X}{\forall x.\, x\inl l_X}$)
from \emph{weak} finiteness in $\Prop$ (i.e.\ $\exists l_X:\List\,X.\,\forall x.\, x\inl l_X$), the list $l_X$ being required
computable in the strong case.
Strong finiteness implies weak finiteness but the converse holds only in restricted contexts, e.g.\ inside proofs of
propositions in $\Prop$.

\smallskip

For the finite PHP, the typical classical proof requires the discreteness
of $X$ to design transpositions/permutations. Here we avoid discreteness
completely, the existence of a duplicate being established without
actually computing one.

\setCoqFilename{php}
\begin{theorem}[Finite PHP][PHP_rel]%
\label{thm:finite_php}
Let $R:X\to Y\to\Prop$ be a binary relation
and $l:\List\,X$ and $m:\List\,Y$ be two lists where
$m$ is shorter than $l$ $(\clen m < \clen l)$.
If $R$ is total from $l$ to $m$
$(\forall x.\, x\inl l \to \exists y.\, y\inl m \land R\,x\,y)$ then
the values at two distinct positions in $l$ are related to the same $y$ in $m$,
i.e.\ there exist $x_1,x_2\inl l$ and $y\inl m$ such that $l$ has shape
$l=\cdots\capp x_1\ccons \cdots\capp x_2\ccons \cdots$ and $R\,x_1\,y$ and $R\,x_2\,y$.
\end{theorem}

\begin{proof}
We start with the case where $R$ is the identity relation $=_X$ on $X$, hence
we want to establish that $l$ contains a duplicate.
 We first prove the following \coqlink[length_le_and_incl_implies_dup_or_perm]{generalised statement}:
if $\clen m\le\clen l$ and $l\incl m$ (i.e.\ $\forall x.\, x\inl l \to x\inl m$)
then either $l$ contains a duplicate
or $l$ and $m$ are permutable. We establish the generalised statement
by structural induction on $m$.

In particular, when $\clen m<\clen l$ then $l$ and $m$ cannot be permutable
(because permutations preserve length), hence \coqlink[finite_pigeon_hole]{$l$ must contain a duplicate}.
Generalizing from $=_X$ to an arbitrary relation $R:X\to Y\to\Prop$ is then a simple
exercise.
\end{proof}

Using the PHP, given a strict order\footnote{i.e.\ an irreflexive and transitive binary relation.} over a finite type $X$,
any descending chain has length bounded
by the size of $X$ as measured by the length of the list enumerating  $X$.

\setBaseUrl{{http://www.ps.uni-saarland.de/extras/fol-trakh-ext/website/Undecidability.}}
\setCoqFilename{Shared.Libs.DLW.Wf.wf_finite}
\begin{myfact}[][wf_strict_order_finite]%
	\label{fact:wf}
	Every strict order on a finite type is well-founded.
\end{myfact}

\begin{proof}
	For a constructive proof, one can for instance show that descending chains cannot contain a duplicate
        (otherwise this would give an impossible cycle in a strict order), hence by the PHP,
        the length of descending chains is bounded by the length of the enumerating list of the finite type.
\end{proof}

Coq's type theory does not provide quotients in general (see e.g.~\cite{cohen13}) but one can build
computable quotients in certain conditions, here for a decidable equivalence
relation of which representatives of equivalence classes are listable.

\setCoqFilename{Shared.Libs.DLW.Utils.fin_quotient}
\begin{theorem}[Finite decidable quotient][decidable_EQUIV_fin_quotient]%
\label{thm:fin_quotient}
Let ${\sim}:X\to X\to\Prop$ be a decidable equivalence with
$\SigType{l_r:\List\,X}{\forall x\exists y.\,y\inl l_r\land x\sim y}$, i.e.\
finitely many equivalence classes.\footnote{Hence $l_r$ denotes a list of \emph{representatives} of equivalence classes.}\enspace%
Then one can compute the quotient $X/{\sim}$ onto $\Fin n$ for some $n$, i.e.
$n:\Nat$, $c:X\to\Fin n$ and $r:\Fin n\to X$ s.t.
$\forall p.\,c\,(r\,p) = p$ and $\forall x y.\,x\sim y\toot c\,x = c\, y$.
\end{theorem}

\begin{proof}
From the list $l_r$ of representatives of equivalence classes, remove duplicate representatives
using the strong decidability of $\sim$. This gives a list $l'_r$ which now contains exactly
one representative for each equivalence class. Convert $l'_r$ to a vector $\vec v$. The function $r$
(representative) is defined by $r\cdef \lambda p.\,\vec v_p$. The function $c$ (for class) is
simple search: $c\,x$ is the first (and unique) $p$ such that $\vec v_p \sim x$.
\end{proof}

Using Theorem~\ref{thm:fin_quotient} with identity over $X$ as equivalence,
we get bijections between finite, discrete types and the type family
$(\Fin n)_{n:\Nat}$.\footnote{For a given $X$, the value $n$ (usually called cardinal) is unique
by the PHP.}

\setCoqFilename{Shared.Libs.DLW.Utils.fin_bij}
\begin{corollary}[][finite_t_discrete_bij_t_pos]%
\label{coro:fin_type}
If $X$ is a finite and discrete type then one can compute~$n:\Nat$
and a bijection from $X$ to $\Fin n$.
\end{corollary}

We conclude this section with the question of the finiteness of powersets.
There is no notion of finite powerset in type theory: even for the unit/singleton
type $\Unit$, there is no list enumerating the predicates in $\Unit\to\Prop$,
even up to extensional equivalence\rlap.\footnote{As there is no list extensionally
enumerating $\Prop$ itself.}\enspace However, there is a notion of \emph{weak powerset}.
Recall that a predicate $p:X\to\Prop$ is \emph{weakly decidable} if it satisfies
$\forall x.\, p\,x\lor\lnot p\,x$, its specialized instance of
the Law of Excluded Middle (LEM)\rlap.\footnote{The general LEM
($\forall P:\Prop.\, P\lor\lnot P$) cannot be established constructively.}\enspace
Assuming $X$ is finite, one can compute a list containing all the weakly decidable
predicates in $X\to \Prop$.

\setBaseUrl{{http://www.ps.uni-saarland.de/extras/fol-trakh-ext/website/Undecidability.Shared.Libs.DLW.Utils.}}
\setCoqFilename{fin_upto}
\begin{lemma}[Weak powerset][finite_t_weak_dec_powerset]%
\label{appendix:lem:wposet}
For every finite type $X$, one can compute a list $ll:\List(X\to\Prop)$ which contains every
weakly decidable predicate in $X\to\Prop$ up to extensional equivalence, i.e.\ $ll$ satisfies
$\forall p:X\to\Prop.\,(\forall x.\, p\,x\lor\lnot p\,x)\to\exists q.\, q\inl ll\land \forall x.\, p\,x\toot q\,x$.
\end{lemma}

\begin{proof}
The list $ll$ is built by induction on the list $l_X:\List\,X$ enumerating $X$.
If $l_X$ is $\cnil$ then $X$ is a void type and thus $ll \cdef (\lambda z.\top)\ccons\cnil$ fits.
If $l_X$ is $x\ccons l$ then we apply the induction hypothesis to $l$ and
get $ll$ for the finite sub-type composed of the elements of $l$ and
we define $ll'\cdef (\lambda p\,z.\, x\neq z\land p\,z)\map ll\capp
                     (\lambda p\,z.\, x= z\lor p\,z)\map ll$. We
check that $ll'$ contains every weakly decidable predicate over $x\ccons l$.
Notice that $\clen{ll'}=2\clen{ll}$ in the induction step, hence one could easily
show that \coqlink[finite_t_weak_dec_powerset]{$\clen{ll}=2^{\clen{l_X}}$}, recovering the cardinality of the
(classical) powerset.
\end{proof}

Notice that while the weak powerset contains all weakly decidable predicates,
it may contain predicates which are not weakly decidable (unless $X$ is
moreover weakly discrete).

\subsection{Representing First-Order Logic}

\setBaseUrl{{http://www.ps.uni-saarland.de/extras/fol-trakh-ext/website/Undecidability.TRAKHTENBROT.}}

We briefly outline our representation of the syntax and semantics of first-order logic in constructive type theory ({cf.}~\cite{forster2021completeness,kirst:2001:synthetic}).
Concerning the \emph{syntax}, we describe terms and formulas as dependent inductive types over a
\setCoqFilename{fo_sig}\coqlink[fo_signature]{signature} $\Sigma=(\Funcs; \Preds)$ of function
symbols $f:\Funcs$ and relation symbols $P:\Preds$ with arities $\arity f$ and $\arity P$,
using binary  connectives ${\binop}\in\{{\dto},{\dand},{\dor}\}$ and
quantifiers ${\quant}\in\{{\dforall},{\dexists}\}$:
		\[\begin{array}{r@{\,:\,}l@{~\bnfdef~}l@{\qquad}l}
                  t &
                  \setCoqFilename{fo_terms}
                  \coqlink[fo_term]{\Term_\Sigma} & x \mid f\,\vec{t}  &  (x:\Nat,~ f:\Funcs,~  \vec t:\Term^{\arity f}_\Sigma\,) \\
		\phi,\psi &
                  \setCoqFilename{fo_logic}
                  \coqlink[fol_form]{\Formula_\Sigma} & \dbot\mid P\,\vec{t}\,
                                   \mid \phi \binop \psi
                                   \mid \quant\phi
                   & (P:\Preds,~ \vec t:\Term^{\arity P}_\Sigma\,)
                 \end{array}\]
	Negation is defined as the abbreviation $\dneg\phi\cdef \phi\,\dto\,\dbot$.

\setCoqFilename{fo_logic}

In the chosen de Bruijn representation~\cite{de_bruijn_lambda_1972}, a bound variable is encoded as the number of
quantifiers shadowing its binder, e.g.\ $\forall x.\,\exists y.\, P\,x\,u\to P\,y\,v$ may be represented
by $\dforall\,\dexists\,P\,1\,4\,\dto\, P\,0\,5$.
The variables $2 = 4-2$ and $3 = 5-2$ in this example are the \emph{free} variables, and variables that do not occur freely are called \emph{fresh},
e.g.\ $0$ and $1$ are fresh.
For the sake of legibility, we write concrete formulas with named binders and defer de Bruijn representations to the Coq development.
For a formula $\phi$ over a signature $\Sigma$, we define
the list $\FV(\phi):\List\,\Nat$ of \coqlink[fol_vars]{free variables},
the list $\funcs\phi : \List\,\Funcs$ of \coqlink[fol_syms]{function symbols}
and the list $\preds\phi : \List\,\Preds$ of \coqlink[fol_rels]{relation symbols}
that actually occur in $\phi$, all by recursion on $\phi$.
We say that $\phi$ is closed if it is void of free variables, i.e.\ $\FV(\phi)=\cnil$.

Turning to \emph{semantics}, we employ the standard (Tarski-style) model-theoretic semantics, evaluating terms in a given domain and embedding the logical connectives into the constructive meta-logic ({cf.}~\cite{veldman_models}):

\setBaseUrl{{http://www.ps.uni-saarland.de/extras/fol-trakh-ext/website/Undecidability.TRAKHTENBROT.}}
\setCoqFilename{fo_sig}
\begin{definition}
	A \emph{\coqlink[fo_model]{model}} $\MM$ over a domain $D:\Type$ is described by a pair of functions
        $\forall f.\,D^{|f|}\to D$ and $\forall P.\,D^{|P|}\to \Prop$ denoted by $f^\MM$ and $P^\MM$.
        \setCoqFilename{fo_terms}%
	Given a \emph{variable assignment} $\rho:\Nat\to D$, we recursively extend it to a \emph{\coqlink[fo_term_sem]{term evaluation}} $\hat\rho:\Term\to D$ with $\hat\rho \,x\cdef \rho\,x$ and $\hat\rho\,(f\,\vec v)\cdef f^\MM\,(\hat{\rho}\map\vec v)$,
        \setCoqFilename{fo_logic}%
	and to the \emph{\coqlink[fol_sem]{satisfaction}} relation $\MM\vDash_\rho \phi$ by
	\begin{align*}
	\MM\vDash_\rho \dot{\bot}&~\cdef ~\bot&
	\MM\vDash_\rho \phi\binop\psi &~\cdef ~\MM\vDash_\rho\phi~\binopm~ \MM\vDash_\rho\psi\\
	\MM\vDash_\rho P\,\vec t\,&~\cdef ~P^\MM\,(\hat{\rho}\map\vec t\,)&
	\MM\vDash_\rho\quant\phi&~\cdef ~\quantm a:D.\,\MM\vDash_{a\cdot\rho} \phi
	\end{align*}
 	where each logical connective ${\binop}/{\quant}$ is mapped to its meta-level counterpart
        ${\binopm}/{\quantm}$ and
        \setCoqFilename{notations}%
        where we denote by $a\cdot\rho$ the \emph{\coqlink[de_bruijn_ext]{de Bruijn extension}} of $\rho$ by $a$, defined by
        $(a\cdot\rho)\,0\cdef a$ and $(a\cdot\rho)\,(1+x)\cdef\rho\,x$.\footnote{The notation $a\cdot\rho$
        illustrates that $a$ is pushed ahead of the sequence $\rho_0,\rho_1,\ldots$}
\end{definition}

A \emph{$\Sigma$-model} is thus a dependent triple $(D,\MM,\rho)$ composed of a domain $D$, a
model $\MM$ for $\Sigma$ over $D$ and an assignment $\rho:\Nat\to D$. It is \emph{finite} if $D$ is finite, and
\emph{decidable} if $P^\MM:D^{|P|}\to \Prop$ is decidable for all $P:\Preds$.

\setCoqFilename{fo_logic}
\begin{myfact}[][fol_sem_dec]%
\label{fact:satisfaction_dec}
Satisfaction
$\lambda \phi.\,\MM\vDash_\rho \phi$ is decidable for finite, decidable $\Sigma$-models.
\end{myfact}

\begin{proof}
By induction on $\phi$; finite quantification preserves decidability.
\end{proof}

In this paper, we are mostly concerned with \emph{finite satisfiability} of formulas.
However, since some of the compound reductions hold for more general or more specific notions,
we introduce the following variants:

\setCoqFilename{fo_sat}
\begin{definition}[Satisfiability]
	For a formula $\phi$ over a signature $\Sigma$, we write
	\begin{itemize}
		\item $\SAT(\Sigma)\,\phi$ if there is a $\Sigma$-model $(D,\MM,\rho)$ such that $\MM\vDash_\rho \phi$;
		\item \coqlink[fo_form_fin_dec_SAT]{$\FSAT(\Sigma)\,\phi$} if additionally $D$ is finite and $\MM$ is decidable;
		\item \coqlink[fo_form_fin_dec_eq_SAT]{$\FSATEQ(\Sigma;\equiv)\,\phi$} if the signature contains a distinguished binary relation symbol $\equiv$
                interpreted as equality, i.e. $x\equiv^\MM y\toot x=y$ for all $x,y:D$.
	\end{itemize}
\end{definition}

\noindent
Notice that in a classical treatment of finite model theory, models are supposed to be given \emph{in extension},
i.e.\ understood as tables providing computational access to functions and relations values.
To enable this view in our constructive setting, we restrict to decidable relations in the definition of $\FSAT$, and from now on,
\emph{finite satisfiability is always meant to encompass a decidable model}.
One could further require the domain $D$ to be discrete to conform more closely with the classical view;
discreteness is in fact enforced by \FSATEQ.
However, we refrain from this requirement and instead show in Section~\ref{sec:discrete_models} that $\FSAT$
and $\FSAT$ over discrete models are constructively equivalent.

\section{Trakhtenbrot's Theorem for a Custom Signature}%
\label{sec:trakh_prelim}

In this section, we show that \BPCP reduces to $\FSATEQ(\SBPCP;{\equiv})$ for the special purpose
signature $\SBPCP\cdef (\{\star^0, e^0, f_\btrue^1, f_\bfalse^1\}; \{P^2, {\prec^2},{\equiv^2}\})$.%
\footnote{We use superscripts to succinctly describe the arity of each symbol.}
To this end, we fix an instance $R:\List\,(\List\,\Bool\times \List\,\Bool)$ of $\BPCP$ (to be understood
as a finite set of pairs of Boolean strings) and
we construct a formula $\phi_R$ such that $\phi_R$ is finitely satisfiable if and only if $R$ has a solution.

Informally, we axiomatise a family $\BB_n$ of  models over the domain of Boolean strings of length bounded by $n$ and let $\phi_R$ express that $R$ has a solution in $\BB_n$.
The axioms express enough equations and inversions of the constructions included in the definition of $\BPCP$ such that a solution for $R$ can be recovered.

Formally, the symbols in $\SBPCP$ are used as follows:
the functions $f_b$ and the constant $e$ represent $b\ccons(\cdot)$ and $\cnil$
for the encoding of strings $s$ as terms $\overline s$ (before a term $\tau$, possibly $e$):
\[\overline{\cnil}\tapp \tau\cdef  \tau\hspace{3em}
\overline{b\ccons s}\tapp \tau\cdef f_b\,(\overline s\tapp \tau)\hspace{3em}
\overline s\cdef  \overline s\tapp e\]
The constant $\star$ represents an undefined value for strings too long to be encoded in the finite model $\BB_n$.
The relation $P$ represents derivability from $R$ (denoted $R\deriv\cdot/\cdot$ here) while $\prec$ and $\equiv$ represent strict suffixes and equality, respectively.

Expected properties of the intended interpretation can be captured formally as first-order formulas.
First, we ensure that $P$ is proper (only subject to defined values) and that $\prec$ is a strict order (irreflexive and transitive):
\[
\begin{array}{c@{~\cdef~}l@{\qquad}l}
\phi_P     & \dforall x y.\, P\,x\,y ~\dto~ x\not\equiv \star ~\dand~ y \not\equiv\star & \text{($P$  proper)} \\
\phi_\prec & (\dforall x.\, x\not\prec x)~\dand~ (\dforall x y z.\, x\prec y~\dto~ y\prec z~\dto~ x\prec z) & \text{($\prec$ strict order)}\\
\end{array}\]
Next, the image of $f_b$ is forced disjoint from $e$ and injective, as long as $\star$ is not reached.
We also ensure that the images of $f_\btrue$ and $f_\bfalse$ intersect only at $\star$:
\[\phi_f~\cdef~\left(\begin{array}{@{\,}l@{\,}}
  f_\btrue\,\star \equiv \star ~\dand~ f_\bfalse\,\star \equiv \star\\
  \dforall x.\, f_\btrue\,x\not\equiv e\\
  \dforall x.\, f_\bfalse\,x\not\equiv e \\
\end{array}\right)
\,\dand\,
\left(\begin{array}{@{\,}l@{\,}}
  \dforall x y.\,f_\btrue\,x\not\equiv \star~\dto~ f_\btrue\,x\equiv f_\btrue\,y~\dto~ x\equiv y \\
  \dforall x y.\,f_\bfalse\,x\not\equiv \star~\dto~ f_\bfalse\,x\equiv f_\bfalse\,y~\dto~ x\equiv y \\
  \dforall x y.\,f_\btrue \,x\equiv f_\bfalse \,y~\dto~ f_\btrue \,x\equiv \star ~\dand~ f_\bfalse \,y\equiv \star\\
\end{array}\right)\]
Furthermore, we enforce that $P$ simulates $R\deriv\cdot/\cdot$, encoding its inversion principle
	\[\phi_\deriv\cdef \dforall x y.\, P\,x\,y ~\dto\,\dbigvee^{\boldsymbol .}_{\clap{\scriptsize$ s/t\!\inl\! R$}}\,
        \dor\left\{\begin{array}{@{\,}l} x\equiv  \overline s ~\dand~ y \equiv \overline t\\
	                                \dexists u v.\,P\,u\,v~\dand~ x\equiv\overline s\tapp u~\dand~ y\equiv\overline t\tapp v~\dand~ u/v\prec x/y\\
                   \end{array}\right.\]
where $u/v\prec x/y$ denotes $(u\prec x~\dot\land~ v\equiv y)\dot\lor (v\prec y~\dot\land~ u\equiv x)\dot\lor (u\prec x ~\dot\land~ v\prec y)$.
Finally, $\phi_R$ is the conjunction of all axioms plus the existence of a solution:
\[\phi_R \cdef  \phi_P ~\dot\land~ \phi_\prec ~\dot\land~ \phi_f ~\dot\land~ \phi_\deriv ~\dot\land~ \dot\exists x.\, P\,x\,x.\]

\setCoqFilename{red_undec}
\begin{theorem}[][BPCP_FIN_DEC_EQ_SAT]%
	\label{thm:BPCP_FSATEQ}
	$\BPCP\red \FSATEQ(\SBPCP;{\equiv})$.
\end{theorem}

\begin{proof}
The reduction $\lambda R.\,\phi_R$ is proved correct by
Lemmas~\ref{lemma:BCPC_FSATEQ} and~\ref{lemma:FSATEQ_BPCP}.
\end{proof}

\setCoqFilename{BPCP_SigBPCP}
\begin{lemma}[][Sig_bpcp_encode_sound]%
        \label{lemma:BCPC_FSATEQ}
	$\BPCP\,R\to \FSATEQ(\SBPCP;{\equiv})\,\phi_R$.
\end{lemma}

\begin{proof}
	Assume $R\deriv s/s$ holds for a string $s$ with $|s|= n$.
	We show that the model $\BB_n$ over Boolean strings bounded by $n$ satisfies $\phi_R$.
	To be more precise, we choose $D_n\cdef \Opt \SigType{s:\List\Bool}{\clen s\le n}$
	as domain, i.e.\ values in $D_n$ are either an (overflow) value $\none$ or a (defined) dependent pair
        $\some{(s,H_s)}$ where $H_s:\clen s\le n$. We interpret the function and relation symbols of the chosen signature by
	\begin{align*}
	e^{\BB_n}&\cdef \cnil & f_b^{\BB_n}\,\none&\cdef  \none & P^{\BB_n}\,s\,t&\cdef R\deriv s/t\\
	\star^{\BB_n}&\cdef \none & f_b^{\BB_n}\,s&\cdef \textnormal{if $\clen s< n$ then $b\ccons s$ else $\none$} &
	s\prec^{\BB_n} t&\cdef s\not = t\land \exists u.\,u\capp s=t
	\end{align*}
	where we left out some explicit constructors and the edge cases of the relations for
        better readability, see the \coqlink[dummy]{Coq code} for full detail.
	As required, $\BB_n$ interprets $\equiv$ by equality $=_{D_n}$.

	Considering the desired properties of $\BB_n$, first note that $D_n$ can be shown finite by induction
        on $n$. This however crucially relies on the proof irrelevance of the $\lambda x.\,x\le n$ predicate\rlap.\footnote{i.e.\ that for every $x:\Nat$ and $H,H':x\le n$ we have $H=H'$. In general, it is
        not always possible to establish finiteness of $\SigType x{P\,x}$ if $P$ is not proof irrelevant.}\enspace
	The atoms $s\prec^{\BB_n} t$ and $s\equiv^{\BB_n} t$ are decidable by straightforward computations
        on Boolean strings. Decidability of $P^{\BB_n}s\,t$ (i.e.\ $R\deriv s/t$) was established in Fact~\ref{fact:BPCP_undec}.
	Finally, since $\phi_R$ is a closed formula, any variable assignment
        $\rho$ can be chosen to establish that $\BB_n$ satisfies $\phi_R$, for instance $\rho\cdef\lambda x.\none$.
	Then showing $\BB_n\vDash_\rho\phi_R$ consists of verifying simple properties of the chosen functions and relations,
        with mostly straightforward proofs.
\end{proof}

\begin{lemma}[][Sig_bpcp_encode_complete]%
        \label{lemma:FSATEQ_BPCP}
	$\FSATEQ(\SBPCP;{\equiv})\,\phi_R\to \BPCP\,R$.
\end{lemma}

\begin{proof}
	Suppose that $\MM\vDash_\rho\phi_R$ holds for some finite $\SBPCP$-model $(D,\MM,\rho)$ interpreting $\equiv$ as equality
        and providing operations $f^\MM_b$, $e^\MM$, $\star^\MM$, $P^\MM$ and $\prec^\MM$.
	Again, the concrete assignment $\rho$ is irrelevant and $\MM\vDash_\rho\phi_R$ ensures that the functions/relations
        behave as specified and that $P^\MM\,x\,x$ holds for some $x:D$.

	Instead of trying to show that $\MM$ is isomorphic to some $\BB_n$, we directly reconstruct a solution for $R$,
        i.e.\ we find some $s$ with $R\deriv s/s$ from the assumption that $\MM\vDash_\rho\phi_R$ holds.
	To this end, we first observe that the relation $u/v\prec^\MM x/y$ as defined above is a strict order and
        thus well-founded as an instance of Fact~\ref{fact:wf}.

	Now we can show that for all $x/y$ with $P^\MM\,x\,y$ there are strings $s$ and $t$ with $x=\overline s$, $y=\overline t$
        and $R\deriv s/t$, by induction on the pair $x/y$ using the well-foundedness of  $\prec^\MM$.
	So let us assume $P^\MM\,x\,y$. Since $\MM$ satisfies $\phi_\deriv$ there are two cases:

	\begin{itemize}
		\item
		there is $s/t\in R$ such that $x=\overline s$ and $y =\overline t$.
		The claim follows by $R\deriv s/t$;
		\item
		there are $u,v:D$ with $P^\MM\,u\,v$ and $s/t\in R$ such that $x=\overline s\tapp u$,  $y=\overline t\tapp v$, and $u/v\prec^\MM x/y$.
		The latter makes the inductive hypothesis applicable for $P^\MM\,u\,v$, hence yielding $R\deriv s'/t'$ for some
                strings $s'$ and $t'$ corresponding to the encodings $u$ and $v$.
		This is enough to conclude $x=\overline{s\capp s'}$, $y=\overline{t\capp t'}$ and $R\deriv(s\capp s')/(t\capp t')$ as wished.
	\end{itemize}

	\noindent
	Applying this fact to the assumed match $P^\MM\,x\,x$ yields a solution $R\deriv s/s$.
\end{proof}

\section{Constructive Finite Model Theory}%
\label{sec:finmod}

\newcommand{\foequiv}[1]{\mathrel{\circeq_{#1}}}
\newcommand{\foindist}{\foequiv{}}
\newcommand{\genrel}{\mathcal R}
\newcommand{\lf}{l_\funcssymb}
\newcommand{\lp}{l_\predssymb}

\newcommand{\Fb}{\mathrm F}
\newcommand{\Ff}{\Fb_\funcssymb}
\newcommand{\Fp}{\Fb_\predssymb}

\newcommand{\bisim}{\mathrel{\equiv_\Fb}}

Combined with Fact~\ref{fact:BPCP_undec}, Theorem~\ref{thm:BPCP_FSATEQ} entails the undecidability (and non-co-enumerability)
of \FSATEQ over a custom (both finite and discrete) signature $\SBPCP$.
By a series of signature reductions, we generalise these results to any signature containing
an at least binary relation symbol. In particular, we explain how to reduce $\FSAT(\Sigma)$
to $\FSAT(\Void; \{\in^2\})$ for any discrete signature $\Sigma$, hence including $\SBPCP$.
We also provide a reduction from $\FSAT(\Void; \{\in^2\})$ to $\FSAT(\{f^n\};\{P^1\})$
for $n\geq 2$, which entails the undecidability of $\FSAT$ for signatures with
one unary relation and an at least binary function.
But first, let us show that $\FSAT$ is unaltered when further assuming
discreteness of the domain.

\subsection{Converting Models to Discrete Ones}%
\label{sec:discrete_models}

We consider the case of models over a discrete domain $D$,
i.e.\ where the equality relation ${=_D}\cdef\lambda x\,y:D.\, x=y$ is decidable.
The question is the following: is $\FSAT$ altered when adding this
further requirement on models.

Of course, in the case of $\FSATEQ(\Sigma;{\equiv})$ the requirement that
$\equiv$ is interpreted as a decidable binary relation which is equivalent to $=_D$
imposes the discreteness of $D$.
But in the case of $\FSAT(\Sigma)$ nothing imposes such a restriction on $D$.
However as we argue below, using Theorem~\ref{thm:fin_quotient},
we can always quotient $D$ along a suitable decidable congruence,
making the quotient a discrete finite type while preserving
first-order satisfaction, from which we deduce that $\FSAT$ is
unaltered by the discreteness requirement; see Section~\ref{sec:fsat_discrete}.

\smallskip

Let us consider a fixed signature $\Sigma=(\Funcs; \Preds)$.
In addition, let us fix a finite type $D$ and a (decidable) model $\MM$
of $\Sigma$ over $D$. Critically, we do not assume the discreteness of $D$.
We can conceive an equivalence over $D$ which is
a congruence for all the interpretations of the symbols
by $\MM$, namely \emph{first-order indistinguishability}
$x\foequiv\Sigma y~\cdef~\forall \phi\,\rho.\, \MM\vDash_{x\cdot\rho}\phi\toot \MM\vDash_{y\cdot\rho}\phi$,
i.e.\ first-order semantics in $\MM$ is not impacted when switching
$x$ with $y$.

The facts that $\foequiv\Sigma$ is both an equivalence and a congruence are easy to prove but,
with this definition, there is little hope of establishing decidability of $\foequiv\Sigma$.
The main reason for this is that the signature may contain symbols of infinitely many arities.
So we further fix two lists $\lf  : \List\,\Funcs$ and $\lp  : \List\,\Preds$ of
function and relation symbols respectively and restrict the congruence
requirement to the symbols in these lists only.

\setCoqFilename{discrete}
\begin{definition}[Bounded first-order indistinguishability][fo_bisimilar]%
\label{def:fo_bisim}
We say that $x$ and $y$ are \emph{first-order indistinguishable up to $\lf /\lp $},
and we write $x\foindist y$, if no first-order formula built from the symbols in $\lf $ and $\lp $ only
can distinguish $x$ from $y$. Formally, this gives:
\[x\foindist y~\cdef~\forall \phi.~\funcs\phi\subseteq\lf\to
\preds\phi\subseteq\lp\to
\forall \rho:\Nat\to D.\,\MM\vDash_{x\cdot\rho}\phi\toot \MM\vDash_{y\cdot\rho}\phi.\]
\end{definition}

To remain simple, we avoid displaying the dependency on $\lf $, $\lp $, $D$ and
$\MM$ in the notation $\foindist$ as they remain fixed in this section anyway.

\smallskip

We claim that first-order indistinguishability $\foindist$ up to $\lf /\lp $ is a strongly decidable
equivalence and a congruence for all the symbols in $\lf /\lp $. Remember that congruence (limited to $\lf /\lp $)
means commutation with the interpretation of symbols in $\MM$:
\[\begin{array}{c}
\forall (f:\Funcs)\,(\vec v\, \vec w:D^{\arity f}).\, f\inl \lf \to (\forall i:\Fin{\arity f}.\, \vec v_i\foindist \vec w_i)
\to f^\MM\,\vec v\foindist f^\MM\,\vec w\\
\forall (P:\Preds)\,(\vec v\, \vec w:D^{\arity P}).\, P\inl \lp \to (\forall i:\Fin{\arity P}.\, \vec v_i\foindist \vec w_i)
\to P^\MM\,\vec v\toot P^\MM\,\vec w.\\
\end{array}\]
We establish the validity of our claim in the following
discussion, ending with Theorem~\ref{thm:fo_indist}.
Equivalence and congruence of $\foindist$  are easy.
However, Definition~\ref{def:fo_bisim} of $\foindist$ hints at no clue for its decidability.
We therefore switch to an alternate definition of $\foindist$
as a bisimulation\rlap.\footnote{That is the greatest fixpoint of an $\omega$-continuous operator.}\enspace
Using Kleene's fixpoint theorem, we would get ${\foindist}$ as $\bigcap_{n<\omega}\Fb^n(\lambda uv.\top)$
for some below defined $\omega$-continuous operator~$\Fb$.
Hopefully, we could ensure that only finitely many (as opposed to $\omega$)
iterations of the operator $\Fb$ are needed for the fixpoint to be
reached, hence preserving finitary properties such as decidability.

\setBaseUrl{{http://www.ps.uni-saarland.de/extras/fol-trakh-ext/website/Undecidability.TRAKHTENBROT.}}
\setCoqFilename{discrete}

\smallskip

So let us define the operators
$\coqlink[fom_op1]{\Ff},\coqlink[fom_op2]{\Fp}:(D\to D\to\Prop)\to(D\to D\to\Prop)$ that
map a binary relation ${\genrel}:D\to D\to\Prop$ to
\[\begin{array}{r@{~\cdef~}l}
\Ff({\genrel}) & \lambda (x\,y : D).\,\forall f.\, f\inl \lf \to \forall(\vec v:D^{\arity f})\,(i:\Fin{\arity f}).\,
   \genrel\,\bigl(f^\MM\,\subst{\vec v\,}xi\bigr)\,\bigl(f^\MM\,\subst{\vec v\,}yi\bigr)\\
\Fp({\genrel}) & \lambda (x\,y : D).\,\forall P.\, P\inl \lp \to \forall(\vec v:D^{\arity P})\,(i:\Fin{\arity P}).\,
   P^\MM\,\subst{\vec v\,}xi \toot P^\MM\,\subst{\vec v\,}yi.\\
\end{array}\]

\begin{myfact}[][discrete_quotient]%
\label{appendix:fact:bisim}
The following results hold for the operator $\Ff$ (resp.\ $\Fp$).
\begin{enumerate}
\coqitem[fom_op_mono] $\Ff$ is monotonic, i.e.\ ${\genrel} \subseteq{\genrel'}\to \Ff({\genrel})\subseteq \Ff({\genrel'})$;
\coqitem[fom_op_continuous] $\Ff$ is $\omega$-continuous, i.e.\ $\bigcap_{n} \Ff({\genrel_n}) \subseteq \Ff\bigl(\bigcap_{n}{\genrel_n}\bigr)$ with decreasing
$(\genrel_n)_{n<\omega}$;
\coqitem[fom_op_id] $\Ff$ preserves reflexivity, i.e. ${=_D}\subseteq \Ff({=_D})$;
\coqitem[fom_op_sym] $\Ff$ preserves symmetry, i.e.\ $\Ff^{-1}({\genrel})\subseteq \Ff({\genrel^{-1}})$;
\coqitem[fom_op_trans] $\Ff$ preserves transitivity, i.e.\ $\Ff({\genrel})\circ \Ff({\genrel})\subseteq \Ff({\genrel}\circ{\genrel})$;
\coqitem[fom_op_dec] $\Ff$ preserves decidability, i.e.\ if ${\genrel}$ is decidable then so is $\Ff({\genrel})$.
\end{enumerate}
Hence the combination $\coqlink[fom_op]{\Fb({\genrel})}\cdef \Ff({\genrel})\cap \Fp({\genrel})$ also preserves these properties.
\end{myfact}

\begin{proof}
The proofs of items (1)--(5) are easy, even without assuming boundedness by $\lf /\lp $. However, 
to ensure  the preservation of decidability (6), that bound is essential for the quantification
over $f$ in $\lf$ (resp.\ $P$ in $\lp$) in the above definition of $\Ff$ (resp.\ $\Fp$) to stay finite.
We observe that since $D$ is finite then so is $D^{\arity f}$ and the remaining quantifications over
$\vec v:D^{\arity f}$ and $i:\Fin{\arity f}$ are finite quantifications again. Hence,
all these quantifications behave as finitary conjunctions and thus preserve decidability. Notice that compared
to $\Ff$, the case of $\Fp$ is degenerated because it does not depend in $\genrel$, hence is
constant.
\end{proof}

As a side remark, notice that one can also show that \coqlink[fol_def_fom_op]{$\Fb$ preserves first-order definability}
where a relation $\genrel:D\to D\to\Prop$ is \emph{first-order definable} if there
is a formula $\phi_\genrel$ built only from $\lf /\lp $ such that
$\forall\rho.\, \genrel\,(\rho\,x_0)\,(\rho\,x_1)\toot  \MM\vDash_\rho \phi_\genrel$.

\begin{theorem}[][fom_eq_fol_characterization]
First-order indistinguishability $\foindist$ up to $\lf /\lp $ is extensionally equivalent
to $\bisim$ (Kleene's greatest fixpoint of $\Fb$), i.e.\ for any $x,y:D$ we have
\[x\foindist y~\toot~x\bisim y\quad\text{where}\quad
{x\bisim y} \,\cdef\, {\forall n:\Nat.\,\Fb^n(\lambda uv.\top)\,x\,y}.\]
\end{theorem}

\begin{proof}
For the $\to$ implication, it is enough to show that $\foindist$ is
a pre-fixpoint of $\Fb$, i.e.\ ${\foindist}\subseteq \Fb({\foindist})$, and
we get this result using \coqlink[fo_bisimilar_fom_eq]{suitable substitutions}. The converse implication
$\leftarrow$ follows from the fact that $\bisim$ is a fixpoint of $\Fb$,
hence it is a congruence for every symbol in $\lf /\lp $, so
$x\bisim y$ entails that \coqlink[fom_eq_fo_bisimilar]{formulas built from $\lf /\lp $ cannot
distinguish $x$ from $y$}.
\end{proof}

With $\bisim$, we now have a more workable characterization of $\foindist$ but still no decidability
result for it since the quantification over $n$ in
$\forall n:\Nat.\, \Fb^n(\lambda uv.\top)\,x\,y$ ranges over the infinite domain $\Nat$.
We now establish that the greatest fixpoint is reached after finitely
many iterations of $\Fb$.
Classically one would argue that $\Fb$ operates over
the finite domain of binary relations over $D$ and since the sequence
$\lambda n.\,\Fb^n(\lambda uv.\top)$ cannot decrease strictly forever (by the PHP),
it must stay constant after at most $n_0\cdef 2^{d\times d}$
iterations where $d\cdef\mathrm{card}\,D$.
Such reasoning is not constructively acceptable as is, due to the impossibility
to build the finite powerset. Fortunately
by Fact~\ref{appendix:fact:bisim} item~6, the iterated values $\Fb^n(\lambda uv.\top)$
are all decidable, hence belong to the weak powerset. As a consequence, we can apply the
finite PHP on the weak powerset.

\setBaseUrl{{http://www.ps.uni-saarland.de/extras/fol-trakh-ext/website/Undecidability.TRAKHTENBROT.}}
\setCoqFilename{discrete}
\begin{theorem}[][fom_eq_finite]%
\label{appendix:thm:bisim_finite}
One can compute $n:\Nat$ such that ${\bisim}$ is equivalent to $\Fb^n(\lambda uv.\top)$.
\end{theorem}

\begin{proof}
\setBaseUrl{{http://www.ps.uni-saarland.de/extras/fol-trakh-ext/website/Undecidability.Shared.Libs.DLW.Utils.}}
\setCoqFilename{fin_upto}
By a \coqlink[finite_t_weak_dec_rels]{variant} of Lemma~\ref{appendix:lem:wposet},
we compute the weak powerset of $D\to D\to\Prop$,\footnote{Via $D\to D\to\Prop\simeq D\times D\to\Prop$,
and finiteness of $D\times D$.} i.e.\
a list $ll$ containing every weakly decidable binary relation over $D$, up to extensional equivalence.
Since $\lambda uv.\top:D\to D\to\Prop$ is strongly decidable
and $\Fb$ preserves (both weak and) strong decidability, the sequence
$\lambda n.\, \Fb^n(\lambda uv.\top)$ is contained in the list $ll$, up to extensional equivalence.
Hence by Theorem~\ref{thm:finite_php} (PHP)\rlap,\footnote{And here we really need a finite PHP
over \emph{non-discrete types}.} after $\clen{ll}$ steps, there must have
been a duplicate, i.e.\ there exists $a<b\le \clen{ll}$ such that $\Fb^a(\lambda uv.\,\top)$
and $\Fb^b(\lambda uv.\top)$ are extensionally equivalent. However the values of $a$ and $b$
are not computed by the PHP but we can still deduce that $\Fb^n(\lambda uv.\top)$
must be stalled after $n=a$, hence a fortiori after $n=\clen{ll}$.
It follows that $\Fb^{\clen{ll}}(\lambda uv.\top)$
is extensionally equivalent to $\bisim$.
\end{proof}

We conclude our construction with the main result: $\foindist$ is a strongly
decidable congruence than can be used to quotient $\MM$ onto a discrete one.

\setCoqFilename{discrete}
\begin{theorem}[][fo_bisimilar_dec_congr]%
\label{thm:fo_indist}
First-order indistinguishability $\foindist$ up to $\lf /\lp $ is a strongly decidable
equivalence and a congruence for all the symbols in $\lf /\lp $.
\end{theorem}

\begin{proof}
Remember that the real difficulty was strong decidability.
By Theorem~\ref{appendix:thm:bisim_finite}, the operator $\Fb$ reaches its fixpoint
$\bisim$ after finitely many steps, and by Fact~\ref{appendix:fact:bisim} item~6, $\Fb$
preserves decidability, hence by an obvious induction, $\bisim$ is decidable.
By Theorem~\ref{coq:fom_eq_fol_characterization}, the equivalent indistinguishability
relation $\foindist$ is decidable.
\end{proof}

Back to our side discussion about first-order definability,
by Theorems~\ref{coq:fom_eq_fol_characterization} and~\ref{appendix:thm:bisim_finite},
there is thus a \coqlink[fom_eq_form_sem]{first-order
formula $\xi$ which characterises first-order indistinguishability up to
$\lf /\lp $ in $\MM$},
i.e.\ $\forall\rho.\, \rho\,x_0 \foindist \rho\,x_1\toot  \MM\vDash_\rho \xi$.
Since its semantics does not depend on variables other that $x_0$ and $x_1$,
one can remap all other variables to e.g.\ $x_0$ hence we can
even ensure that the first-order formula characterising $\foindist$ contains
only two free variables, namely $x_0$ and $x_1$.

\subsection{Removing Model Discreteness and Interpreted Equality}%
\label{sec:fsat_discrete}

We use the strongly decidable congruence $\foindist$ to quotient models onto
discrete ones (in fact $\Fin n$ for some $n$) while preserving first-order satisfaction.

\setCoqFilename{fo_sat}
\begin{definition}[][fo_form_fin_discr_dec_SAT]
	We write $\coqlink[FSAT']{\FSAT'}(\Sigma)\,\phi$ if $\FSAT(\Sigma)\,\phi$ on a discrete model.
\end{definition}

\setCoqFilename{red_utils}
\begin{theorem}[][fo_form_fin_dec_SAT_discr_equiv]%
        \label{thm:quotient_FSAT}
	For every first-order signature $\Sigma$ and formula $\phi$ over $\Sigma$,
        we have $\FSAT(\Sigma)\,\phi$ iff $\FSAT'(\Sigma)\,\phi$,
        and as a consequence, both reductions
        $\FSAT(\Sigma)\red \FSAT'(\Sigma)$ and
        $\FSAT'(\Sigma)\red \FSAT(\Sigma)$ hold.
\end{theorem}

\begin{proof}
$\FSAT(\Sigma)\,\phi$ entails $\FSAT'(\Sigma)\,\phi$ is the non-trivial implication.
Hence we consider a finite $\Sigma$-model $(D,\MM,\rho)$ of $\phi$ and we build a new
finite $\Sigma$-model of $\phi$ which is furthermore discrete.
We collect the symbols occurring in
$\phi$ as the lists $\lf \cdef\funcs\phi$ (for functions) and $\lp \cdef\preds\phi$ (for relations).
By Theorem~\ref{thm:fo_indist}, first-order indistinguishability ${\foindist}:D\to D\to\Prop$
up to $\funcs\phi/\preds\phi$ is  a strongly decidable equivalence over $D$ and a congruence for the semantics of the symbols occurring in $\phi$.
\setCoqFilename{Sig_discrete}%
Using Theorem~\ref{thm:fin_quotient}, we \coqlink[fo_discrete_removal]{build the quotient} $D/{\foindist}$ on a $\Fin n$
for some $n:\Nat$. We transport the model $\MM$ along this quotient
and because $\foindist$ is a congruence for the symbols in $\phi$, its
semantics is preserved along the quotient. Hence, $\phi$ has a finite
model over the domain $\Fin n$ which is both finite and discrete.
\end{proof}

\setCoqFilename{red_utils}
\begin{theorem}[][FIN_DEC_EQ_SAT_FIN_DEC_SAT]%
	\label{thm:uninterpret}
        If $\equiv$ is a binary relation symbol in the
        signature $\Sigma$, one has a
        reduction $\FSATEQ(\Sigma;{\equiv})\red \FSAT(\Sigma)$.
\end{theorem}

\begin{proof}
\setCoqFilename{fo_congruence}
Given a list $\lf$ (resp.\ $\lp$) of function (resp.\ relation) symbols such that ${\equiv}$ belongs to $\lp$,
we construct a formula \coqlink[fol_congruence]{$\psi(\lf,\lp,{\equiv})$} over the function symbols
in \coqlink[fol_congruence_syms]{$\lf $} and relation symbols in \coqlink[fol_congruence_rels]{$\lp$} expressing
the requirement that $\equiv$ is
an \coqlink[fol_sem_congruence]{equivalence and a congruence} for the symbols in $\lf /\lp $.
Then we show that
$\lambda\phi.\, \phi~\dand~\psi(\funcs\phi, {\equiv}::\preds\phi,{\equiv})$
is a correct reduction, where $\funcs\phi$ and $\preds\phi$
list the symbols occurring in $\phi$.
\end{proof}

\subsection{From Discrete Signatures to Singleton Signatures}

Let us start by converting a discrete signature to a finite and discrete
signature.

\setCoqFilename{Sig_Sig_fin}
\begin{lemma}[][Sig_discrete_to_pos]%
\label{lem:signature_finite}
For any formula $\phi$ over a discrete signature $\Sigma$,
one can compute a
signature $\Sigma_{n,m}=(\Fin n;\Fin m)$,
arity preserving maps $\Fin n\to \Funcs$ and $\Fin m\to \Preds$,
and an \emph{equi-satisfiable} formula $\psi$ over  $\Sigma_{n,m}$, i.e.
$\FSAT(\Sigma)\,\phi\toot\FSAT(\Sigma_{n,m})\,\psi$.
\end{lemma}

\begin{proof}
We use the discreteness of $\Sigma$ and bijectively map the lists of symbols $\funcs\phi$ and $\preds\phi$
onto $\Fin n$ and $\Fin m$ respectively, using
Corollary~\ref{coro:fin_type}. We structurally map $\phi$ to $\psi$
over $\Sigma_{n,m}$ along this bijection, which preserves finite satisfiability.
\end{proof}

Notice that $n$ and $m$ in the signature $\Sigma_{n,m}$ depend on $\phi$, hence
the above statement cannot be presented as a reduction between (fixed) signatures.

\smallskip

We now erase all function symbols by encoding them with relation symbols.
To this end, let $\Sigma=(\Funcs; \Preds)$ be a signature, we set $\Sigma'\cdef(\Void; \{\equiv^2\}+\Funcs^{+1}+\Preds)$
where $\equiv$ is a new interpreted relation symbol of arity two and in the conversion, function symbols have arity lifted by
one, hence the $\Funcs^{+1}$ notation.

\setCoqFilename{red_undec}
\begin{lemma}[][FIN_DISCR_DEC_SAT_FIN_DEC_EQ_NOSYMS_SAT]%
	\label{lem:remove_functions}
        For any finite\footnote{In the Coq code, we prove the theorem for finite \emph{or} discrete types of function symbols.}
        type of function symbols $\Funcs$, one can construct a
        reduction $\FSAT'(\Funcs; \Preds)\red \FSATEQ(\Void; \{\equiv^2\}+\Funcs^{+1}+\Preds;{\equiv^2})$.
\end{lemma}

\begin{proof}
The idea is to  recursively replace a term $t$ over $\Sigma$ by a formula which is ``equivalent'' to
$x\equiv t$ (where $x$ is a fresh variable not occurring in $t$)
and then an atomic formula like e.g.\ $P\,[t_1;t_2]$ by  $\exists\, x_1\,x_2.\, x_1\equiv t_1\,\dand\,x_2\equiv t_2\,\dand\, P\,[x_1;x_2]$.
We complete the encoding with a formula stating that every function symbol $f:\Funcs$ is encoded into a
total functional relation $P_f: \Funcs^{+1}$ of arity augmented by $1$.

Notice that \setCoqFilename{utils}%
\coqlink[graph_tot_reif]{constructively recovering actual functions from total functional relations}
is made possible by the finiteness/discreteness of the domain, combined with the
decidability of the semantic interpretation of the corresponding relation symbols.
In particular, this would fail in the case of SAT with potentially infinite models.
\end{proof}

Next, assuming that the function symbols have already been erased, we explain how to merge the relation symbols in a signature
$\Sigma=(\Void; \Preds)$ into a single relation symbol, provided that there is an upper bound for the arities in $\Preds$.

\begin{lemma}[][FSAT_REL_BOUNDED_ONE_REL]%
\label{lem:bounded_arity}
The reduction $\FSAT(\Void; \Preds)\red\FSAT\bigl(\Void; \{Q^{1+n}\}\bigr)$ holds
when $\Preds$ is a finite and discrete type of relation symbols and
$|P|\le n$ holds for all $P:\Preds$.
\end{lemma}

\begin{proof}
This comprises three independent reductions, see Fact~\ref{prop:three_reds} below.
\end{proof}

In the following, we denote
by $\Funcs^n$ (resp.\ $\Preds^n$) the same type of function (resp.\
relation) symbols but where the arity is uniformly converted to $n$.

\begin{myfact}%
  \label{prop:three_reds}
  Let $\Sigma=(\Funcs; \Preds)$ be a signature:
  \begin{enumerate}
  \coqitem[FSAT_UNIFORM] $\FSAT(\Funcs; \Preds)\red \FSAT(\Funcs; \Preds^n)$ if $|P|\le n$ holds for all $P:\Preds$;
  \coqitem[FSAT_ONE_REL] $\FSAT(\Void; \Preds^n)\red \FSAT(\Preds^0;\{Q^{1+n}\})$ if $\Preds$ is finite;
  \coqitem[FSAT_NOCST] $\FSAT(\Funcs^0; \Preds)\red \FSAT(\Void; \Preds)$ if $\Funcs$ is discrete.
  \end{enumerate}
\end{myfact}

\begin{proof}
For the first reduction, every atomic formula of the form $P\,\vec v$ with $\clen{\vec v\,}=\arity P\leq n$ is converted to $P\,\vec w$
with $\vec w\cdef \vec v \capp [x_0;\dots;x_0]$ and $\clen{\vec w\,}= n$ for an arbitrary term variable $x_0$.
The rest of the structure of formulas is unchanged.

For the second reduction, we convert every atomic formula $P\,\vec v$ with $\clen{\vec v\,}=n$
into $Q(P\ccons\vec v)$ where $P$ now represents a constant symbol
($Q$ is fixed).

For the last reduction, we replace every constant symbol by a corresponding fresh variable
chosen above all the free variables of the transformed formula.
\end{proof}

\subsection{\mathversion{bold}%
  Compressing \texorpdfstring{$n$}{n}-ary Relations to Binary Membership}%
\label{sec:compressing}

Let $\Sigma_n=(\Void;\{P^n\})$ be a singleton signature where $P$ is of arity $n$.
We now show that $P$ can be compressed to a binary relation modelling
set membership via a construction using hereditarily finite
sets~\cite{SmolkaStark:2016:Hereditarily}
(useful only when $n\ge 3$).

\begin{theorem}[][FIN_DISCR_DEC_nSAT_FIN_DEC_2SAT]%
	\label{thm:compress_relations}
	$\FSAT'(\Void; \{P^n\})\red \FSAT(\Void; \{\din^2\})$.
\end{theorem}

\newcommand{\pair}[2]{({#1},{#2})}
\newcommand{\bpair}[2]{\bigl({#1},{#2}\bigr)}
\newcommand{\myred}{\Sigma_{n\rightsquigarrow 2}}
\newcommand{\myredr}{\myred'}
\newcommand{\powerset}{\mathcal P}
\newcommand{\dequal}{\mathrel{\dot\approx}}
\newcommand{\dis}{\mathrel{\dot\equiv}}

Technically, this reduction is one of the most involved in this work, although in most
presentations of Trakhtenbrot's theorem, this is left as an ``easy exercise\rlap,'' see 
e.g.~\cite{Libkin:2010:EFM:1965351}. Maybe it is perceived so because it relies
on the encoding of tuples in classical set theory, which is somehow natural for
mathematicians, but
properly building the finite set model
in constructive type theory was not that easy.

Here we only give an overview
of the main tools.
We encode an arbitrary $n$-ary relation $R:X^n\to\Prop$ over a finite and discrete type $X$
in the theory of \emph{membership} over the signature $\Sigma_2=(\Void; \{{\din}^2\})$.
Membership is much weaker than set theory because the only required
set-theoretic axiom is \emph{extensionality}: it suffices to prove the
properties of the below described encoding of pairs, ordered pairs and $n$-tuples.
Two sets are extensionally equal if
their members are the same, and we define identity as:
\[x\dequal y {~\cdef~} \dforall z.\, z\din x~\dtoot~z\din y\]
Extensionality states that two extensionally equal sets
belong to the same sets:
\begin{equation}%
\label{eq:ext}
\dforall x y.\, x\dequal y~\dto~\dforall z.\, x\din z~\dto~y\din z
\end{equation}
As a consequence of the extensionality axiom~\eqref{eq:ext}, no first-order formula over $\Sigma_2$ can distinguish two extensionally
equal sets\footnote{Even if $x$ and $y$ have the same elements, this does not imply their identity
in every model. However they are identical in the quotient model of Section~\ref{sec:discrete_models},
because then $x$ and $y$ are first-order indistinguishable.} because identity $\dequal$ is
a congruence for membership~$\din$. Moreover,
establishing the identity $x\dequal y$ of two sets reduces to proving that they have the same elements.

\smallskip

Notice that the language of membership theory (and set theory) does not contain
any function symbol, hence, contrary to usual mathematical practices, there is no
other way to handle a set than via a characterising formula which makes it
a very cumbersome language to work with formally. However, this is how we have to proceed
in the Coq development but here, we also use to meta-level terms in the prose for simplicity.

\setCoqFilename{membership}

Following Kuratowski, the ordered pair of two sets $x$ and $y$ is encoded as $\pair x y = \bigl\{\{x\},\{x,y\}\bigr\}$.
In the first-order theory of membership as implemented in Coq, with terms limited to variables, this means
we encode the sentences ``$p\dis \{x,y\}$'' and
``$p\dis \pair x y$''  by:\footnote{where the $\dis$ notation reads as ``represents\rlap.''}
\[\begin{array}{r@{~\cdef~}l}
 \coqlink[mb_is_pair]{p \dis\{x,y\}}    & \dforall a.\, a\din p~\dtoot~ a \dequal x \,\dor\, a\dequal y\\[0.5ex]
 \coqlink[mb_is_opair]{p \dis \pair x y} & \dexists a b.\, a\dis \{x,x\} ~\dand~ b \dis \{x,y\} ~\dand~ p \dis \{a,b\}\\
\end{array}\]
while the $n$-tuple/vector $(x_1,\ldots,x_n)$ is encoded as $\bpair{x_1}{(x_2,\ldots,x_n)}$
recursively. Hence we encode the sentence ``\coqlink[mb_is_tuple]{$t \dis\vec v$\,}'' recursively on $\vec v\,$ by:
\[
t\dis\cnil {~\cdef~} \dforall a.\, \dneg\,a\din t \qquad\qquad
t\dis x\ccons \vec v\,  {~\cdef~} \dexists p.\, t\dis \pair {x} p~\dand~p\dis \vec v
\]
Finally we encode ``tuple $\vec v\,$ belongs to $r$'' as
 $\coqlink[mb_is_tuple_in]{\vec v\din r}\cdef \dexists t.\,t\dis \vec v~\dand~t\din r$.

\setCoqFilename{Sign_Sig2}

We can now describe the reduction function which maps formulas  over $\Sigma_n$ to
formulas over $\Sigma_2$. We reserve two first-order
variables $d$ (for the domain $D$) and $r$ (for the relation $R$). We describe
the recursive part of the reduction \coqlink[Sign_Sig2_encoding]{$\myredr$}:
\[\begin{array}{@{}r@{~\cdef~}l@{\qquad}r@{~\cdef~}l@{}}
  \myredr(P\,\vec v) & \vec v\din r
& \myredr(\dforall z.\,\phi) & \dforall z.\,z\din d~\dto~\myredr(\phi)\\
  \myredr(\phi\,\binop\,\psi) & \myredr(\phi)\,\binop\,\myredr(\psi)
& \myredr(\dexists z.\,\phi) & \dexists z.\,z\din d~\dand~\myredr(\phi)\\
\end{array}
\]
ignoring the de Bruijn syntax (which would imply adding $d$ and $r$ as
extra parameters of $\myredr$). Notice that we have to prevent $d$ and $r$ from occurring
freely in $\phi$.
In addition, for given $\phi$ we set:
\[
\begin{array}{r@{~\cdef~}l@{\qquad}l}
\phi_1 & \dexists z.\, z\din d        & \text{i.e.\ $d$ is non-empty;}\\[0.5ex]
\phi_2 & x_1\din d~\dand~\cdots~\dand~x_k\din d & \text{where $[x_1;\ldots;x_k] = \FV(\phi)$.}
\end{array}\]
This gives us the reduction function $\myred(\phi) \cdef \phi_1\,\dand\,\phi_2\,\dand\,\myredr(\phi)$.

\smallskip

The \emph{\coqlink[SAT2_SATn]{completeness}}
of the reduction $\myred$ is the easy part. Given a finite model of $\myred(\phi)$
over $\Sigma_2$, we recover a model of $\phi$ over $\Sigma_n$ by selecting as the new
domain the members of $d$,\footnote{using the Boolean counterpart of the decidable membership predicate to
get a finite dependent sum.}
and the interpretation of $P\,\vec v$ is given by testing whether the $n$-tuple $\vec v$ belongs to $r$.

The \emph{\coqlink[SATn_SAT2]{soundness}} of the reduction $\myred$ is the formally involved part, with Theorem~\ref{thm:Sign_Sig2}
below containing the key construction.

\setCoqFilename{reln_hfs}
\begin{theorem}[][reln_hfs]%
\label{thm:Sign_Sig2}
Given a decidable $n$-ary relation $R:X^n\to\Prop$ over a finite, discrete
and inhabited type $X$, one can compute a finite type $Y$ equipped
with a decidable relation ${\in}:Y\to Y\to\Prop$, two distinguished
elements $d,r:Y$ and a pair of maps $i:X\to Y$ and $s:Y\to X$ such that:
\begin{enumerate}[~~1.]
\item for any $x:X$, $i\,x\in d$ holds;
\item for any $y\in d$, there exists $x$ such that $y=i\,x$;
\item for any $\vec v:X^n$, $R\,\vec v$ holds iff ``tuple $i(\vec v\,)$ belongs to $r$\rlap.''
\end{enumerate}
\end{theorem}

\begin{proof}
	We give a brief outline of this proof, referring to the \coqlink[reln_hfs]{Coq code} for details.
	The type $Y$ is built from the type of hereditarily finite sets based on~\cite{SmolkaStark:2016:Hereditarily},
	and when we use the word ``set'' below, it means hereditarily finite set. There, sets are finitely branching ordered trees
        (encoded as binary trees in the standard way), considered up to permutation and contraction equivalence.
        \setCoqFilename{btree}%
        By \coqlink[bt_lt_eq_lt_dec]{totally ordering those trees}, a \coqlink[bt_norm]{normal form} can be computed to give an effective representative for every equivalence class,
        \setCoqFilename{hfs}%
        implementing the quotient on this infinite type of trees in Coq, which gives us \coqlink[hfs]{hereditarily finite sets}.

        Basing on (hereditarily finite) sets,
	the idea behind the construction of $Y$ is first to construct $d$ as a \coqlink[hfs_transitive]{\emph{transitive set}}\footnote{Recall that the members of
        a transitive set are also subsets of it.} of which the elements are in bijection $i/s$
        \setCoqFilename{reln_hfs}%
	with the type $X$, hence $d$ is the \coqlink[X_surj_hfs]{cardinal} of $X$ in the set-theoretic meaning.
	Then the iterated powersets $\powerset(d), \powerset^2(d),\ldots,\powerset^k(d)$ are all transitive sets
        as well and contain $d$
	both as a member and as a subset. Considering $\powerset^{2n}(d)$
	which contains all the $n$-tuples built from the members of $d$, we
	define $r$ as the set of $n$-tuples collecting the encodings $i(\vec v\,)$ of vectors
	$\vec v:X^n$ such that $R\,\vec v$. We show $r\in p$ for $p$ defined as
	$p\cdef \powerset^{2n+1}(d)$.
        Using the Boolean counterpart of $(\cdot)\in p$ for unicity of proofs,
        we then define \coqlink[Y]{$Y\cdef \{z\mid z\in p\}$},
	restrict membership $\in$ to $Y$ and this gives the finite type equipped
	with all the required properties. Notice that the decidability requirement for ${\in}$
        holds constructively because we work with hereditarily finite sets, and would not hold with arbitrary sets.
\end{proof}

Notice that while the axiom of extensionality for membership~\eqref{eq:ext} is absent from
the reduction function $\myred$, it is however satisfied in the model of hereditary finite sets
we build to prove Theorem~\ref{thm:Sign_Sig2}, and instrumental in the set-theoretic encoding of
the $n$-ary relation implemented there.

\subsection{Summary: From Discrete Signatures to the Binary Signature}

Combining all of the previous results,
we give a reduction from any discrete signature
to the binary singleton signature.

\setCoqFilename{red_undec}
\begin{theorem}[][DISCRETE_TO_BINARY]%
\label{thm:discrete_to_binary}
$\FSAT(\Sigma)\red \FSAT(\Void; \{P^2\})$ holds for any discrete signature $\Sigma$.
\end{theorem}

\begin{proof}
Let us first consider the case of $\Sigma_{n,m}=(\Fin n;\Fin m)$, a signature
over the finite and discrete types $\Fin n$ and $\Fin m$.
Then we have a reduction $\FSAT(\Fin n;\Fin m)\red\FSAT(\Void; \{P^2\})$
by combining Theorems~\ref{thm:quotient_FSAT},~\ref{thm:uninterpret}, and~\ref{thm:compress_relations}
and Lemmas~\ref{lem:remove_functions} and~\ref{lem:bounded_arity}.

Let us denote by $f_{n,m}$ the reduction $\FSAT(\Fin n;\Fin m)\red\FSAT(\Void; \{P^2\})$.
Let us now consider a fixed discrete signature $\Sigma$. For a formula $\phi$ over $\Sigma$,
using Lemma~\ref{lem:signature_finite}, we compute a signature $\Sigma_{n,m}$
and formula $\psi_\phi$ over $\Sigma_{n,m}$ s.t.\ $\FSAT(\Sigma)\,\phi\toot\FSAT(\Fin n;\Fin m)\,\psi$.
The map $\lambda\phi.f_{n,m}\,\psi_\phi$ is the required reduction.
\end{proof}

The binary signature can be further reduced to a signature with a non-monadic function.

\begin{lemma}[][FSAT_REL2_to_FUNnREL1]%
	\label{lem:sig2_sig21}
	$\FSAT(\Void; \{P^2\})\red \FSAT(\{f^n\};\{Q^1\})$ when $n\geq 2$.
\end{lemma}

\setCoqFilename{Sig2_SigSSn1}

\begin{proof}
We encode the binary relation $\lambda x\,y.\,P\,[x;y]$ with $\lambda x\,y.\,Q \bigl(f\,[x;y;\dots]\bigr)$,
using the first two parameters of $f$ to encode pairing. But since we need to change the
domain of the model, we also use a fresh variable $d$ to encode the domain as
$\lambda x.\,Q (f\,[d;x;\dots])$ and we \coqlink[Sig2_SigSSn1_encoding]{restrict all quantifications to the domain}
similarly to the encoding $\myredr$ of Section~\ref{sec:compressing}.
\end{proof}

We finish the reduction chains with the weakest possible signature constraints. The
following reductions have straightforward proofs.

\setCoqFilename{red_undec}
\begin{myfact}%
	\label{prop:embed_sig}
        One has reductions for the three statements below (for $n\ge 2$):
        \begin{enumerate}
        \coqitem[FSAT_REL_2ton] $\FSAT(\Void; \{P^2\})\red \FSAT(\Void; \{P^n\})$;
	\coqitem[FSAT_RELn_ANY] $\FSAT(\Void; \{P^n\})\red \FSAT(\Sigma)$ if $\Sigma$ contains an $n$-ary relation (symbol);
        \coqitem[FSAT_FUNnREL1_ANY] $\FSAT(\{f^n\};\{Q^1\})\red \FSAT(\Sigma)$ if $\Sigma$ contains an $n$-ary function and a unary relation.
        \end{enumerate}
\end{myfact}

\section{Decidability Results}%
\label{sec:decidability}

Complementing the previously studied negative results, we now examine the conditions allowing for decidable satisfiability problems.
Since any binary symbol renders finite satisfiability undecidable (see upcoming Theorem~\ref{thm:full_trakhtenbrot}), we only need to consider monadic signatures where all arities
are below one. First we consider the case of discrete mononadic signatures, and then give the most general characterisation
of decidability of finite satisfiability, by showing (in the monadic case) that it is equivalent to the decidability
of the Boolean discernability of function and relation symbols.

\subsection{Discrete Monadic Signatures}

\begin{definition}[Monadic]
A first-order signature is \emph{monadic} if the arities of all symbols are either 0 or 1, functions and relations alike.
A first-order signature is \emph{degenerate/propositional} if the arities of all relation symbols is 0.
\end{definition}

Notice that when relation symbols have arity $0$, then no term can occur in first-order formulas, hence the
degenerate or propositional qualifier.

\setCoqFilename{fo_sat_dec}
\begin{lemma}[FSAT over a fixed domain][FSAT_in_dec]%
\label{lem:FSAT_in}
Given a discrete signature $\Sigma$ and a discrete and
finite type $D$, one can decide whether or not a formula
over $\Sigma$ has a (finite) model over domain $D$.
\end{lemma}

\begin{proof}
By Fact~\ref{fact:satisfaction_dec},
satisfaction in a given finite model
is decidable.
It is also \coqlink[FO_model_equiv_spec]{invariant under extensional equivalence}, so we only need to show that there are
\coqlink[finite_t_model_upto]{finitely many (decidable) models}
over $D$ up to extensional equivalence\rlap.%
\footnote{Without discreteness of $\Sigma$, it is impossible
to build the list of models over $D=\Bool$.}
\end{proof}

\begin{lemma}[][fo_form_fin_discr_dec_SAT_pos]%
\label{lem:FSAT_in_pos}
A formula over a signature $\Sigma$
has a finite and discrete model if and only if it has a (finite) model over
$\Fin n$ for some $n:\Nat$.
\end{lemma}

\begin{proof}
If $\phi$ has a model over
a discrete and finite domain $D$, by Corollary~\ref{coro:fin_type},
one can bijectively map $D$ to $\Fin n$  and transport
the model along this bijection.
\end{proof}

\setCoqFilename{red_dec}
\begin{lemma}[][FSAT_MONADIC_DEC]
\label{lem:monadic_fol}
$\FSAT(\Void;\Preds)$ is decidable if $\Preds$ is discrete with uniform arity $1$.
\end{lemma}

\begin{proof}
\setCoqFilename{Sig1}
By Lemma~\ref{lem:signature_finite}, we can assume $\Preds=\Fin n$ for some $n:\Nat$ w.l.o.g.\enspace
We show that if $\phi$ has a finite model then it must have
\coqlink[Monadic_model_bounded]{a model over the domain $\{\vec v:\Bool^n\mid b\,\vec v = \btrue\}$} for
some Boolean subset $b:\Bool^n\to\Bool$.
Up to extensional equivalence, there are only finitely many such subsets
$b$ and we conclude with Lemma~\ref{lem:FSAT_in}.
\end{proof}

The following Lemma~\ref{lemma:remove_funcs} explains how to remove all function symbols
but proceeds entirely inside monadic signatures.
Contrary to Lemma~\ref{lem:remove_functions} that encodes $n$-ary functions
with $(1{+}n)$-ary (functional) relations, here we cannot increase the arity by one because we
intend the target signature to be monadic also.

\begin{lemma}[][FSAT_MONADIC_11_FSAT_MONADIC_1]%
\label{lemma:remove_funcs}
Let $n:\Nat$, and $\Preds$ be a finite type of relation symbols. For \emph{signatures of uniform arity $1$},
we have a reduction $\FSAT(\Fin n;\Preds)\red\FSAT(\Void;(\List\,{\Fin n}\times\Preds)+\Preds)$.
\end{lemma}

\begin{proof}
\setCoqFilename{Sig1_1}
We implement a proof somewhat inspired by that of Proposition~6.2.7 (Grädel)
in~\cite[pp.~251]{borger1997classical}
but the invariant suggested in the iterative process described there did not work
out formally
and we had to proceed in a single conversion step instead, switching
from single symbols to lists of function symbols.

\smallskip

Given a formula potentially containing function symbols in the finite type $\Fin n$, the problem is
to remove those functions symbols while preserving \FSAT. In the monadic
signature $(\Fin n;\Preds)$ (with uniform arity $1$), an atomic formula
is necessarily of the form \coqlink[fot_word_var_eq]{$P(f_1(\ldots f_q(x)\ldots))$} where
$P:\Preds$, $[f_1;\ldots;f_q]:\List{\Fin n}$
is a list of function symbols, and $x:\Nat$ is a variable.
The removal of function symbols proceeds in two steps: in a first phase, by adding more predicate symbols, we
reduce to the case where logical atoms are of the form $P(f(x))$ or $P(x)$, i.e.\ at most one nested function
application; and in a second phase, we remove the function symbols by skolemisation.

\smallskip

In the first step, we assume an upper bound $m$ over the number $q$ of successive applications
to be found in an atomic formula, e.g.\ $P(f_1(\ldots f_q(x)\ldots))$. This number $m$
is \emph{computed globally} from the start
formula to be converted, e.g.\ by choosing the \coqlink[max_depth]{maximum depth of nesting} of function symbols occurring in it.
For disambiguation, we denote the relation symbols of the target signature either by $P_r$ with $r:\Preds$
(those that originate in $\Preds$) or newly the introduced ones by $Q_{w,r}$ with  $(w,r):\List\,{\Fin n}\times\Preds$.

To the atom $P_r(f_1(\ldots f_q(x)\ldots))$ we \coqlink[encode]{associate} the new atom
$Q_{[f_q;\ldots; f_1],r}(x)$. To ensure the correctness of this encoding, we add the
equivalences
\[ Q_{\cnil,r}(x)~\dtoot~ P_r(x) \qquad \text{and}\qquad Q_{f\ccons w,r}(x)~\dtoot~Q_{w,r}(f(x))\]
all universally quantified over $x$. When satisfied, these equivalences uniquely characterise the
interpretation of the predicates $Q_{w,r}$ and ensure that the replacing atom $Q_{[f_q;\ldots; f_1],r}(x)$
has the same interpretation as the original $P_r(f_1(\ldots f_q(x)\ldots))$.

Notice however that there are infinitely many words $w$ in $\List\,{\Fin n}$ (as soon as $n>0$),
and as a consequence, infinitely many instances of $Q_{f\ccons w,r}(x)~\dtoot~Q_{w,r}(f(x))$.
As these would not all be embeddable a single first-order formula, we use the bound $m$
to limit the equivalences to the instances where $\clen w<m$, of which there
are indeed only finitely many.

\smallskip

In the second step (skolemisation), we simply convert the equivalences
\[Q_{f\ccons w,r}(x)~\dtoot~Q_{w,r}(f(x))
\qquad\text{into}\qquad
Q_{f\ccons w,r}(x)~\dtoot~Q_{w,r}(x_f)\]
where the $x_f$ are $n$ new skolem variables (as many as in the type $\Fin n$).
To preserve \FSAT, we need to existentially quantify over the skolem variables $x_0,\ldots,x_{n-1}$,
these existential quantifications $\exists x_0\dots\exists x_{n-1}$
occurring above the conjunction of all the equations $Q_{f\ccons w,r}(x)~\dtoot~Q_{w,r}(x_f)$.
On top of this, we keep universal quantification over $x$.
\end{proof}

If functions or relations have arity $0$, one can always
lift them to arity $1$ by filling the hole with an arbitrary term,
like in~Fact~\ref{prop:three_reds}, item~(1).

\begin{myfact}[][FSAT_FULL_MONADIC_FSAT_11]%
\label{prop:rem_cst_props}
The reduction $\FSAT(\Funcs;\Preds)\red\FSAT(\Funcs^1;\Preds^1)$ holds when
all arities in $\Sigma$ are at most 1,
where $\Funcs^1$ and $\Preds^1$ denote arities uniformly updated to $1$.
\end{myfact}

\setCoqFilename{red_dec}
\begin{myfact}%
\label{fact:propositional_inter_reductible}
Let $\Sigma=(\Funcs;\Preds^0)$ be a signature where the arities of all relation symbols is $0$.
Then $\FSAT(\Sigma)$ and $\FSAT(\Void,\Preds^0)$ are inter-reducible, i.e.\ we
have both $\coqlink[FSAT_PROP_FSAT_x0]{\FSAT(\Sigma)\red\FSAT(\Void;\Preds^0)}$ and
$\coqlink[FSAT_x0_FSAT_PROP]{\FSAT(\Void;\Preds^0)\red\FSAT(\Sigma)}$.
\end{myfact}

\begin{proof}
If all relation symbols are constants (arity $0$), then we contemplate
a degenerate/propositional fragment of first-order logic. Indeed, regardless of $\Funcs$,
no term can occur in formulas, hence neither can function symbols.
Since relations are the same with the same arity $0$ in both signatures,
there are obvious $\FSAT$ preserving structural transformations
(back and forth) of formulas of these degenerate first-order signatures.
\end{proof}

\setCoqFilename{red_dec}
\begin{theorem}[][FULL_MONADIC]%
	\label{thm:monadic_fol}
	$\FSAT(\Sigma)$ is decidable if $\Sigma$ is discrete with arities less or equal than $1$, or
	if all relation symbols have arity $0$.
\end{theorem}

\begin{proof}
	If all arities are at most $1$, then by Fact~\ref{prop:rem_cst_props}, we can assume $\Sigma$ of uniform arity~$1$.
	Therefore, for a formula $\phi$ over $\Sigma$ with uniform arity $1$, we
	need to decide $\FSAT$ for $\phi$. By Theorem~\ref{lem:signature_finite}, we can compute a
	signature $\Sigma_{n,m}=(\Fin n;\Fin m)$ and a formula $\psi$ over $\Sigma_{n,m}$
	equi-satisfiable with $\phi$.
	Using the reduction of Lemma~\ref{lemma:remove_funcs}, we compute a formula $\gamma$,
	equi-satisfiable with $\psi$, over a discrete signature of uniform arity $1$, void of functions.
	We decide the satisfiability of $\gamma$ by Lemma~\ref{lem:monadic_fol}.

	On the other hand, if all relation symbols have arity $0$, we use the reduction
        $\FSAT(\Sigma)\red\FSAT(\Void,\Preds^0)$ of Fact~\ref{fact:propositional_inter_reductible}
        and we are back to the previous case.
\end{proof}

\subsection{Discernability and Finite Satisfiability}

\newcommand{\dsfun}{\delta}
\newcommand{\dscrn}{\mathrel{\not\equiv_d}}
\newcommand{\udscrn}{\mathrel{\equiv_d}}

We say that two members $x$ and $y$ of a type $X$ (of e.g.\ symbols) are \emph{Boolean discernable}
if there is a Boolean valued map $\dsfun:X\to\Bool$ giving them two different values, i.e.\ $\dsfun\,x \neq \dsfun\,y$.
In the sequel, we simply say discernable or undiscernable for conciseness. We however recall the reader not
to confuse (Boolean) undiscernability with Leibniz's understanding of undiscernability\footnote{$x$ is
Leibniz-undiscernable from $y$ if for any predicate $P:X\to\Prop$, $P(x)$ implies $P(y)$.}
which agrees with ``equality'' in the case of Coq. Formally this gives the definition:

\setCoqFilename{discernable}
\begin{definition}[Discernability]
Given a type $X$ and two terms $x,y:X$, we say that $x$ and $y$
are \coqlink[discernable]{\emph{discernable}} and write $x\dscrn y$ if $\exists \dsfun:X\to\Bool.\,\dsfun\,x \neq \dsfun\,y$.
On the other hand, we say that $x$ and $y$ are \coqlink[undiscernable]{\emph{undiscernable}}
and we write $x\udscrn y$ if $\forall \dsfun:X\to\Bool.\,\dsfun\,x=\dsfun\,y$.
\end{definition}

\coqlink[discernable_equiv1]{Equivalently}, $x$ and $y$ are discernable if they can be mapped to two different values of a discrete type. Hence
$\dsfun$ discerns $x$ from $y$ and the decision algorithm for identity between $\dsfun\,x$ and $\dsfun\,y$ eventually
produces a proof of discernability.  Because $\Bool$ is discrete,
undiscernability is the negation of discernability,
i.e.\ \coqlink[undiscernable_spec]{$\forall xy:X.\, x\udscrn y\toot \neg (x\dscrn y)$}.

Undiscernability is an equivalence relation generally weaker than equality but the two notions match on discrete types.
As we explain below, it is the proper notion to capture when two symbols of a first-order signature
cannot be interpreted differently by any finite first-order model. The reason is that
those models are inherently discrete:
by Theorem~\ref{thm:quotient_FSAT}, the interpretation
of function symbols can be restricted to discrete models;
and relation symbols are always interpreted by decidable relations.

\setCoqFilename{Sig_discernable}
\begin{myfact}[][FSAT_dec_implies_discernable_rels_dec]%
\label{fact:FSAT_discern_rels}
Let $\Sigma=(\Funcs;\Preds)$ be a signature such that
$\FSAT(\Sigma)$ is decidable.
Then for any two relation symbols $P$ and $Q$ in $\Preds$, it is decidable
whether $P$ and $Q$ are discernable.
\end{myfact}

\begin{proof}
The formula $P({\cdots})\,\dand\, \dneg Q({\cdots})$ is finitely satisfiable if
and only if $P\dscrn Q$. The dots $({\cdots})$ are arbitrary
terms filling the mandatory arguments of $P$ and $Q$.
\end{proof}

\begin{myfact}[][FSAT_dec_implies_discernable_syms_dec]%
\label{fact:FSAT_discern_syms}
Let $\Sigma=(\Funcs;\Preds)$ be a signature with $P:\Preds$, a relation
symbol of arity 1. If
$\FSAT(\Sigma)$ is decidable,
then for any two function symbols $f$ and $g$ in $\Funcs$, it is decidable
whether $f$ and $g$ are discernable.
\end{myfact}

\begin{proof}
The formula $P\bigl(f({\cdots})\bigr)\,\dand\, \dneg P\bigl(g({\cdots})\bigr)$ is finitely satisfiable if
and only if $f\dscrn g$.
\end{proof}

We restrict the discussion that follows to monadic signatures only
because otherwise $\FSAT$ is undecidable
(see upcoming Theorem~\ref{thm:full_trakhtenbrot}). Hence function and
relation symbols have arity either $0$ or $1$.
If a monadic signature contains a unary relation symbol, then decidability
of $\FSAT$ entails decidability of discernability of both function symbols and relation
symbols. On the other hand, if a monadic signature contains no unary relation (degenerate case), then
all relations are constant (zero-ary) and in this case we can only obtain an algorithm
for discerning relation symbols.

\begin{lemma}[][Sig_discernable_dec_to_discrete]
Let $\Sigma^1=(\Funcs^1;\Preds^1)$ be a signature with uniform arity 1 such
that both $\Funcs$ and $\Preds$ have decidable discernability. For any
formula $\phi$ over $\Sigma^1$ one can compute a discrete signature
$\Sigma_d^1$ of uniform arity 1 and a formula $\psi$ over $\Sigma_d^1$
such that $\phi$ and $\psi$ are equi-satisfiable, i.e.
$\FSAT(\Sigma^1)\,\phi\toot\FSAT(\Sigma_d^1)\,\psi$.
\end{lemma}

\begin{proof}
Let $S$ be a type of symbols that has decidable discernability (either $\Funcs$ or $\Preds$). Consequently,
\coqlink[discernable_dec_undiscernable_dec]{undiscernability is also decidable}. Since undiscernable symbols cannot be distinguished by first-order
models, we would like to quotient along undiscernability but it is not always possible to
do so constructively\rlap.\footnote{e.g.\ if there is no available enumeration of $\udscrn$-equivalence classes.}\enspace However, using a variant of Theorem~\ref{thm:fin_quotient},
given a list $l_S:\List S$ of symbols, we \coqlink[discriminable_list]{compute} a discrete type\footnote{Notice that $D$ is also finite by construction but this property is not needed below.}
$D$ and a map $\dsfun:S\to D$
such that
\[\forall rs.\, r\inl l_S\to s\inl l_S\to (r\udscrn s \toot \dsfun\,r=\dsfun\,s).\]
In a way,
$\dsfun$ is able to discern symbols simultaneously, but only those in $l_S$. We can now use $\dsfun$ to project
the symbols of $S$ on a discrete type while preserving $\FSAT$ for formulas that contain symbols
in $l_S$ only.

Given a start formula $\phi$ in the signature $\Sigma^1$ (of uniform arity $1$) with decidable discernability of
symbols, we compute the function and relation symbols occurring in $\phi$ and use those two lists
to compute two discrete types $\funcs d$, $\preds d$ and two maps $\dsfun_{\funcs{}}:\Funcs\to \funcs d$, $\dsfun_{\preds{}}:\Preds\to \preds d$
identifying the undiscernable symbols occurring in $\phi$. Using $\dsfun_{\funcs{}}$ and $\dsfun_{\preds{}}$, we map the formula
$\phi$ on a formula $\psi$ in the discrete signature $\Sigma_d^1=(\preds d^1;\funcs d^1)$ (of uniform arity $1$ also)
while preserving $\FSAT$, i.e.\ $\FSAT(\Sigma^1)\,\phi\toot\FSAT(\Sigma_d^1)\,\psi$.
Notice that the target signature depends on (the symbols occurring in) $\phi$.
\end{proof}

\setCoqFilename{red_dec}
\begin{theorem}[][FSAT_FULL_MONADIC_discernable]%
\label{thm:FSAT_full_monadic_discern}
Let $\Sigma=(\Funcs;\Preds)$ be a monadic signature (all arities are less than 1)
such that both $\Funcs$ and $\Preds$ have decidable discernability. Then
$\FSAT(\Sigma)$ is decidable.
\end{theorem}

\begin{proof}
By Fact~\ref{prop:rem_cst_props}, to show the $\FSAT(\Sigma)$ is decidable we may assume
that $\Sigma$ is of uniform arity $1$, i.e.\ of the form $\Sigma^1=(\Funcs^1;\Preds^1)$.
Given $\phi$ over $\Sigma^1$, we compute an equi-satisfiable $\psi$ over the discrete
signature $\Sigma_d^1$
with Lemma~\ref{coq:Sig_discernable_dec_to_discrete},
and then decide $\FSAT(\Sigma_d^1)\,\psi$ using Theorem~\ref{thm:monadic_fol}.
\end{proof}

\begin{theorem}[][FSAT_PROP_ONLY_discernable]%
\label{thm:FSAT_prop_discern}
Let $\Sigma=(\Funcs;\Preds)$ be a propositional monadic signature (all arities of relation
symbols are 0) such that $\Preds$ has decidable discernability. Then
$\FSAT(\Sigma)$ is decidable.
\end{theorem}

\begin{proof}
We combine the reduction $\FSAT(\Sigma)\red\FSAT(\Void;\Preds^0)$
of Fact~\ref{fact:propositional_inter_reductible} and Theorem~\ref{thm:FSAT_full_monadic_discern}.
\end{proof}

\section{Final Signature Classification}%
\label{sec:trakh_full}

We conclude with the exact classification of $\FSAT$ regarding enumerability, decidability, and undecidability depending on the properties of the signature.

\setCoqFilename{red_enum}
\begin{theorem}[][FSAT_opt_enum_t]%
\label{thm:FSAT_enum}
Given $\Sigma=(\Funcs;\Preds)$ where both $\Funcs$
and $\Preds$ are data types, the finite satisfiability problem
for formulas over $\Sigma$ is  enumerable.
\end{theorem}

\begin{proof}
Using Theorem~\ref{thm:quotient_FSAT} and Lemmas~\ref{lem:FSAT_in} and~\ref{lem:FSAT_in_pos}, one constructs
a predicate $Q:\Nat\to \Formula_\Sigma\to\Bool$ s.t.\
\coqlink[FSAT_rec_enum_t]{$\FSAT(\Sigma)\,\phi\toot\exists n.\,Q\,n\,\phi = \btrue$}.
Then, it is easy to build a \coqlink[FSAT_opt_enum_t]{computable enumeration} $e:\Nat\to\Opt\,\Formula_\Sigma$ of $\FSAT(\Sigma):\Formula_\Sigma\to\Prop$.
\end{proof}

\setCoqFilename{red_dec}
\begin{theorem}[Full Monadic FOL][FULL_MONADIC_discernable]%
	\label{thm:full_monadic_fol}
	$\FSAT(\Sigma)$ is decidable if either (a) or (b) holds:
\begin{description}
\item[(a)] $\Sigma$ is monadic and has decidable discernability for both function and relation symbols;
\item[(b)] relation symbols in $\Sigma$ have all arity $0$ together with decidable discernability.
\end{description}
\end{theorem}

\noindent
This simply recollects Theorems~\ref{thm:FSAT_full_monadic_discern} and~\ref{thm:FSAT_prop_discern}.
However, we point out that in the case of monadic signatures,
because of Facts~\ref{fact:FSAT_discern_rels} and~\ref{fact:FSAT_discern_syms}, it is not
possible to weaken further the above assumptions (a) or (b) to entail the decidability of $\FSAT(\Sigma)$.
And for non-monadic signatures, we have the general statement of Trakhtenbrot's theorem:

\setCoqFilename{red_undec}
\begin{theorem}[Full Trakhtenbrot][FULL_TRAKHTENBROT]%
\label{thm:full_trakhtenbrot}
        $\BPCP$ reduces to $\FSAT(\Sigma)$ if either (a) or (b) holds:
\begin{description}
\item[(a)] $\Sigma$ contains either an at least binary relation symbol;
\item[(b)] $\Sigma$ contains a unary relation symbol together with an at least binary function symbol.
\end{description}
\end{theorem}

\begin{proof}
	By Theorems~\ref{thm:BPCP_FSATEQ},~\ref{thm:uninterpret} and~\ref{thm:discrete_to_binary},
        Lemma~\ref{lem:sig2_sig21}, and Fact~\ref{prop:embed_sig}.
\end{proof}

According to the explanation provided in \Cref{sec:undec} we may (informally) observe that this verified reduction yields undecidability:

\begin{observation}%
	\label{larry:trakh}
	For a signature $\Sigma$ satisfying the conditions (a) or (b) of Theorem~\ref{thm:full_trakhtenbrot}, $\FSAT(\Sigma)$ is undecidable and, more specifically, not co-enumerable.
\end{observation}

Notice that even if the conditions on arities of Theorems~\ref{thm:full_monadic_fol}
and~\ref{thm:full_trakhtenbrot} fully classify signatures with decidable discernability, \textit{a fortiori}
discrete signatures, it is not possible to decide whether a signature is monadic or not,
unless the signature is furthermore finite.
For a given formula $\phi$ though, it is always possible to render it in the finite signature of used symbols.

\section{Application: Undecidability of Separation Logic}%
\label{sec:seplog}

\setBaseUrl{{http://www.ps.uni-saarland.de/extras/fol-trakh-ext/website/Undecidability.}}
\setCoqFilename{SeparationLogic.SL}

Trakhtenbrot's theorem can be used to establish several negative results concerning first-order logic and other decision problems.
Regarding first-order logic, a direct consequence is that the completeness theorem fails in the sense that no (enumerable) deduction system can capture finite validity, as this would turn $\FSAT(\Sigma)$ co-enumerable in contradiction to \Cref{larry:trakh}.
Examples of decision problems shown undecidable by trivial reduction from $\FSAT(\Sigma)$ are the finite satisfiability problem in the first-order theory of graphs and problems such as query containment and query equivalence in data base theory.
This section is devoted to a still simple but insightful application of Trakhtenbrot's theorem, namely the undecidability of separation logic.

Separation logic~\cite{reynolds2002separation,ishtiaq2001bi} as an assertion language for finite data structures bears an obvious connection to finitely interpreted first-order logic.
In particular the formulation in~\cite{10.1007/3-540-45294-X_10} adding pointers to \emph{binary} heap cells $(t\mapsto t_1,t_2)$ to the spatial operations $\phi *\psi$ and $\phi\sep \psi$ for separating conjunction and implication, as well as $\emp$ for the emptiness assertion in extension of the pure first-order language, admits a compact reduction from the $\FSAT$ problem over the \emph{binary} signature.
In this section, we outline our adaptation of the undecidability proof given in~\cite{10.1007/3-540-45294-X_10} to our constructive and type-theoretic setting, with a focus on the technical details concerning discreteness and decidability induced by this approach.

We represent the syntax of separation logic as an inductive type \coqlink[sp_form]{$\SL$} of formulas $\phi,\psi$ by
\[\phi,\psi:\SL~\bnfdef~ (t\mapsto t_1,t_2)\mid\emp\mid \phi * \psi \mid \phi \sep \psi\mid  t_1\equiv t_2\mid\dbot \mid \phi\binop \psi\mid\quant\phi\hspace{2em}(t:\Opt\,\Nat)\]
with ${\binop}\in\{{\dto},{\dand},{\dor}\}$ and ${\quant}\in\{{\dforall},{\dexists}\}$ as in \Formula and isolate a minimal fragment \setCoqFilename{SeparationLogic.MSL}\coqlink[msp_form]{$\MSL$} by
\[\phi,\psi:\MSL~\bnfdef~ (t\hookrightarrow t_1,t_2)\mid\dbot \mid \phi\binop \psi\mid\quant\phi\hspace{2em}(t:\Opt\,\Nat).\]

Informally, the assertions expressed by $\SL$ and $\MSL$ are interpreted over a memory model consisting of a finite heap addressing binary cells and a stack mapping variables to addresses.
The first-order fragment is interpreted as expected, where quantification ranges over addresses.
The pointer $(t\mapsto t_1,t_2)$ is interpreted strictly as the assertion that the heap consists of a single cell containing the pair denoted by $t_1$ and $t_2$ referenced at $t$, while $(t\hookrightarrow t_1,t_2)$ just asserts that the heap contains such a pair.

\setCoqFilename{SeparationLogic.SL}
\begin{definition}[][sp_sat]
	Given a \emph{stack} $s:\Nat \to\Val$ mapping variables to possibly invalid addresses $\Val :=\Opt\,\Nat$ and a \emph{heap} $h:\List\,(\Nat\times(\Val\times\Val))$ representing a finite map of valid addresses to pairs of addresses, we define the \emph{satisfaction relation} $h\vDash_s\phi$ for $\phi:\SL$ recursively by
	\begin{small}
		\begin{align*}
		h\vDash_s (t\mapsto t_1,t_2)&~:=~\exists a.\, \hat s\,t=\some a \land h=[(a,(\hat s\,t_1,\hat s\,t_2))]&
		h\vDash_s t_1\equiv t_2 &~:=~ \hat s\,t_1=\hat s\,t_2\\
		h\vDash_s \emp &~:=~ h = \cnil&
		h\vDash_s \dbot &~:=~ \bot\\
		h\vDash_s \phi *\psi &~:=~ \exists h_1h_2.\,h\approx h_1\!\!\capp\!\! h_2 \,\land\, h_1\vDash_s \phi \,\land\, h_2\vDash_s \psi&
		h\vDash_s \phi\binop \psi&~:=~ h\vDash_s \phi ~\binopm~ h\vDash_s \psi\\
		h\vDash_s \phi \sep \psi &~:=~ \forall h'.\,h\# h'\,\to\, h'\vDash_s \phi \,\to\, h\!\!\capp\!\! h'\vDash_s \psi&
		h\vDash_s \quant\phi &~:=~ \quantm v:\Val.\,h\vDash_{v\cdot s}\phi
		\end{align*}
		where $\hat s\,\some x:=s\,x$ and $\hat s\,\none:=\none$, where $h\approx h'$ denotes equivalence $\forall a\,p.\,(a,p)\in h \leftrightarrow (a,p)\in h$ and $h\# h'$ denotes disjointness $\neg \exists a\,p\,p'.\,(a,p)\in h \land (a,p')\in h'$,
		and where for the first-order fragment each logical connective ${\binop}/{\quant}$ is mapped to its meta-level counterpart ${\binopm}/{\quantm}$.

		For $\phi:\MSL$ the relation $h\vDash_s\phi$ is obtained by the same rules with additionally
		\[h\vDash_s (t\hookrightarrow t_1,t_2) ~:=~\exists a.\, \hat s\,t=\some a \land (a,(\hat s\,t_1,\hat s\,t_2))\in h\]
		and we define the satisfiability problem \emph{\coqlink[SLSAT]{\SLSAT} (\setCoqFilename{SeparationLogic.MSL}\coqlink[MSLSAT]{\MSLSAT})} on \emph{$\phi:\SL$ ($\phi:\MSL$)} as the existence of a stack $s$ and a functional heap $h$ (i.e. $\forall app'.\,(a,p)\in h \to (a,p')\in h\to p=p'$) such that $h\vDash_s\phi$.
	\end{small}

\end{definition}

The outline of the following reduction is to first establish $\FSAT\preceq \MSLSAT$ to emphasise that already the fragment \MSL is undecidable and then continue with $\MSLSAT\preceq\SLSAT$ by a mere syntax embedding.
The idea for the main reduction is to encode the binary relation $P[x;y]:\Formula$ on the heap by $(a\hookrightarrow x,y):\MSL$ at some address $a$ while tracking the domain elements $x$ via empty cells $(x\hookrightarrow \none,\none)$.
Since the intermediate model transformations require computational access to the domain, we actually need to base the main reduction on $\FSAT'$ guaranteeing discreteness.

\setCoqFilename{SeparationLogic.Reductions.FSATdc_to_MSLSAT}
Formally, we translate first-order formulas $\phi:\Formula$ over the binary signature $(\Void; \{P^2\})$ to formulas \coqlink[encode]{$\overline\phi :\MSL$} in the sufficient fragment of separation logic by
\begin{align*}
	\overline{P[x;y]}&:=(\dexists z.\, (\some z\hookrightarrow \some{x}, \some{y}))~\dand~ (\some x \hookrightarrow \none,\none) ~\dand~(\some y \hookrightarrow \none,\none)\\
	\overline{\dforall x.\,\phi}&:=\dforall x.\,(\some x\hookrightarrow\none,\none) ~\dto~ \overline \phi\hspace{2.5em}
	\overline{\dexists x.\,\phi}:=\dexists x.\,(\some x\hookrightarrow\none,\none) ~\dand~ \overline \phi
\end{align*}
and by recursively descending through the remaining logical operations.
The next two lemmas verify the correctness of the reduction function.

\begin{lemma}[][reduction_forwards]%
	\label{sl1}
	Given a model \MM over a discrete and finite domain $D$ exhausted by a duplicate-free list $l_D$ and coming with a decider $f_P$ for $P^\MM$, one can compute a functional heap $h$ and from every environment $\rho:\Nat \to D$ a stack $s_\rho$ such that $\MM\vDash_\rho\phi~\toot~h\vDash_{s_\rho}\overline \phi$ for all $\rho$.
\end{lemma}

\begin{proof}
	We encode domain elements $d:D$ as natural numbers by the unique index $n_d<\clen{l_D}$ at which $d$ occurs in $l_D$ and pairs $(d,e):D\times D$ as numbers $n_{d,e}\ge \clen{l_D}$ by $n_{d,e}:=\pi(n_ d,n_ e)+\clen{l_D}$ employing an injective pairing function $\pi:(\Nat\times \Nat)\to\Nat$.
	We then construct a heap $h$ encoding both the full domain $D$ and the binary relation $P^\MM$ by
	\[h:=[(n_d,(\none,\none))\mid d\in l_D]\capp [(n_{d,e},(\some{n_ d}, \some{n_ e}))\mid f_P\,d\,e=\btrue]\]
	which is functional by the injectivity of the encodings $n_d$ and $n_{d,e}$ that are taken care not to overlap by the addition of $\clen{l_D}$ in the definition of $n_{d,e}$.

	Given an environment $\rho$, we convert it to a stack $s_\rho\, x := \some{n_{\rho\,x}}$ and prove the claimed equivalence by induction on $\phi$ with $\rho$ generalised.
	We only discuss the case $\phi=P[x;y]$.
	Assuming $\MM\vDash_\rho P[x;y]$ we have $P^\MM[d;e]$ and so $f_P\,d\,e=\btrue$ for $d:=\rho\,x$ and $e:=\rho\,y$.
	Therefore $(n_{d,e},(\some{n_ d}, \some{n_ e}))\in h$.
	Since by construction also $(n_d,(\none,\none))\in h$ and $(n_e,(\none,\none))\in h$, we can conclude $h\vDash_{s_\rho}\overline{P[x;y]}$.
	Conversely, given $h\vDash_{s_\rho}\overline{P[x;y]}$, we know that $(a,(\some{n_ d}, \some{n_ e}))\in h$ at some address $a$ and hence can deduce that $f_P\,d\,e=\btrue$.
	Thus $\MM\vDash_\rho P[x;y]$.
\end{proof}

\begin{lemma}[][reduction_backwards]%
	\label{sl2}
	Given a heap $h$ containing at least one element of the form $(a_0,(\none,\none))$, one can compute a decidable model $\MM$ over a discrete and finite domain $D$ and from every stack $s$ an environment $\rho_s:\Nat\to D$ such that $\MM \vDash_{\rho_s}\phi ~\toot~h\vDash_s \overline \phi$ for all stacks $s$ that satisfy the condition $\forall x\in\FV(\phi).\,\exists a.\, s\,x=\some a\land (a,(\none,\none))\in h$.
\end{lemma}

\begin{proof}
	As domain $D$ we take the type of all addresses $a$ such that $(a,(\none,\none))\in h$, formally defined as $D:=\SigType{a}{(a,(\none,\none))\in h}$
	using the Boolean counterpart of list membership for unicity of proofs.
	By this construction, elements $a$ and $a'$ of $D$ are equal iff they are equal as addresses in $\Nat$ and so $D$ is discrete and, since it is bounded by $h$, also finite.
	To turn $D$ into a model $\MM$, we set
	\[P^\MM[a_1;a_2]:=\exists a.\,(a,(\some{a_1},\some{a_2}))\in h\]
	which is decidable since it expresses a bounded quantification over a list.

	Given the dummy element $a_0$ of $D$, we can convert a stack $s$ into an environment $\rho_s$ mapping $x$ to $a$ if $s\,x=\some a$ with $(a,(\none,\none))\in h$, and to $a_0$ in any other case.
	With these constructions in place, the claim is established by induction on $\phi$ with $s$ generalised, we again just discuss the case $\phi=P[x;y]$.
	First suppose $\MM \vDash_{\rho_s}P[x;y]$, then since $x,y\in\FV(P[x;y])$ we know that $\rho_s$ is well-defined on $x,y$ by the condition on $s$ and obtain corresponding $a_1,a_2:D$.
	From $P^\MM[a_1;a_2]$ we obtain $a$ with $(a,(\some{a_1},\some{a_2}))\in h$ and therefore conclude $h\vDash_s \overline{P[x;y]}$.
	Conversely, starting with $h\vDash_s \overline{P[x;y]}$ straightforwardly yields $\MM \vDash_{\rho_s}P[x;y]$.
\end{proof}

Since \Cref{sl2} imposes a condition on the free variables of the input formula, it is convenient to start from the restriction $\FSAT'_c$ of $\FSAT'$ to closed formulas $\phi$.
To connect $\FSAT'_c$ to \BPCP, one could verify that the reduction chain $\BPCP\preceq \FSAT'$ given in the previous sections actually yields closed formulas, hence witnessing $\BPCP\preceq \FSAT'_c$ directly.
However, we prefer to implement the general (and straightforward) reduction $\FSAT'\preceq\FSAT'_c$.
The following theorem then summarises the three parts comprising the reduction $\FSAT'\preceq \SLSAT$:

\begin{theorem}%
	\label{thm:SL}
	We have reductions as follows:
	\begin{enumerate}
		\setCoqFilename{FOL.Reductions.FSATd_to_FSATdc}
		\coqitem[reduction]
		$\FSAT'(\Sigma)\preceq\FSAT'_c(\Sigma)$
		\setCoqFilename{SeparationLogic.Reductions.FSATdc_to_MSLSAT}
		\coqitem[reduction]
		$\FSAT'_c(\Void; \{P^2\})\preceq \MSLSAT$
		\setCoqFilename{SeparationLogic.Reductions.MSLSAT_to_SLSAT}
		\coqitem[reduction]
		$\MSLSAT\preceq \SLSAT$
	\end{enumerate}
\end{theorem}

\begin{proof}
	We establish each reduction separately.
	\begin{enumerate}[1.]
		\item
		Given a formula $\phi$ over signature $\Sigma$, we first compute a bound $N$ such that $x<N$ for all $x\in \FV(\phi)$.
		Then it is easy to verify that $\phi':=\dexists^N \phi$ prefixing $\phi$ with $N$ existential quantifiers is closed and has a (finite and discrete) model iff $\phi$ does.
		\item
		Given a closed formula $\phi$ over the binary signature $(\Void; \{P^2\})$, we define $\phi':\MSL$ by
		\[\phi':=(\dexists\,(\some 0\hookrightarrow\none,\none))~\dand~ \overline \phi\]
		and show that $\phi$ has a finite and discrete model iff $\phi'$ is $\MSL$-satisfiable.

		First, if $\MM\vDash_\rho\phi$ over a finite and discrete domain $D$, we can apply \Cref{sl1} to obtain $h\vDash_{s_\rho}\overline \phi$ since over discrete types any list can be turned duplicate-free.
		Moreover, $D$ at least contains the element $d:=\rho\,0$ and hence by construction of $h$ we have $(n_d,(\none,\none))\in h$, establishing the guard $\dexists\,(\some 0\hookrightarrow\none,\none)$.
		So in total $h\vDash_{s_\rho}\phi'$.

		Secondly, from $h\vDash_s\phi'$ we obtain $h\vDash_s\overline\phi$ and some $a_0$ with $(a_0,(\none,\none))\in h$.
		Since $\phi$ is closed, the condition in \Cref{sl2} holds vacuously and thus we obtain $\MM \vDash_{\rho_s}\phi$.
		\item
		We embed $\MSL$ into $\SL$ by the map sending the sole deviating assertion $(t\hookrightarrow t_1,t_2)$ to $(t\mapsto t_1,t_2)*\dtop$ where $\dtop:=\dbot\dto\dbot$.
		To verify this reduction, it suffices to establish that $h\vDash_s(t\hookrightarrow t_1,t_2)$ iff $h\vDash_s(t\mapsto t_1,t_2)*\dtop$, which follows by straightforward list manipulation.
		\qedhere
	\end{enumerate}
\end{proof}

\setCoqFilename{SeparationLogic.SL_undec}
\begin{corollary}[][SLSAT_undec]
	$\FSAT'(\Sigma)$ reduces to \SLSAT, therefore also \BPCP reduces to \SLSAT.
\end{corollary}

\begin{proof}
	By composing the three parts of \Cref{thm:SL} with Theorems~\ref{thm:full_trakhtenbrot} and~\ref{thm:quotient_FSAT}.
\end{proof}

As in the case of $\FSAT$ and backed by the explanation in \Cref{sec:undec}, we may interpret this reduction as an undecidability result:

\begin{observation}
\SLSAT is undecidable and, more specifically, not co-enumerable.
\end{observation}

In comparison to~\cite{10.1007/3-540-45294-X_10}, our reduction is formulated for satisfiability problems instead of the dual validity problems.
However, this change is inessential since the models are transformed pointwise as visible in \Cref{sl1} and
\Cref{sl2} and so the only consequence is a flipped quantifier in the proof of (2) of \Cref{thm:SL}.
More importantly, the formal setting forced us to be more explicit about the handling of addresses, in particular the encoding of a given finite first-order interpretation.
For instance, the way chosen in \Cref{sl1} to start with an abstract domain $D$ and encode both elements and pairs over $D$ as numbers is not the only alternative but allowed us to maintain the explicit representation of the address space as $\Nat$.
Moreover, our syntax fragments differ slightly since we don't need equality in $\MSL$ as it is not a primitive in $\Formula$ but on the other hand keep all logical connectives and not just the classically sufficient base ${\dto,\dforall,\dbot}$ as this is (in the general case) constructively insufficient.

We end this section with the remark that the reduction given in~\cite{10.1007/3-540-45294-X_10} and adapted here of course crucially relies on the binary pointers $(t\mapsto t_1,t_2)$ as a language primitive.
As discussed in~\cite{BROCHENIN2012106}, with a less explicit memory structure, the considered fragment of separation logic is decidable and only turns undecidable on addition of separating implication.

\section{Discussion}%
\label{sec:discussion}

\subsection{Code overview, implementation, and design choices}

The main part of our Coq development directly concerned with the
\href{https://github.com/uds-psl/coq-library-undecidability/tree/trakhtenbrot_ijcar/theories/TRAKHTENBROT}{classification of finite satisfiability}
consists of a bit more than 10k loc, in addition to 3k loc of (partly reused)
\href{https://github.com/uds-psl/coq-library-undecidability/tree/trakhtenbrot_ijcar/theories/Shared/Libs/DLW}{utility libraries}.
Most of the code comprises the signature transformations,
with more than 4k~loc for reducing discrete signatures to membership.
Comparatively, the initial reduction from $\BPCP$ to $\FSATEQ(\SBPCP)$ takes less than 500~loc.
The application to the undecidability of separation logic also amounts to roughly 500~loc.

Our mechanisation of first-order logic in principle follows previous developments~\cite{ForsterCPP,ForsterEtAl:2019:Completeness,forster2021completeness} but also differs in a few aspects that were picked up by follow-up work~\cite{kirst:2001:synthetic}.
Notably, we had to separate function from relation signatures to be able to express distinct signatures that agree on one sort of symbols computationally.
This mostly avoids wandering in ``setoid hell\rlap,'' i.e.\ the painful manipulation of 
cumbersome type castings on formul{\ae}, of which the type \emph{depends} on signatures.
Moreover, we found it favorable to abstract over the logical connectives in form of $\binop$ and $\quant$ to shorten
purely structural definitions and proofs. Most of the signature reductions proceed in such structural approaches
so this is a real relief for propositional connectives. For existential quantifiers~$\dexists$ and universal
quantifiers~$\dforall$, it is however less frequent that they can be managed in a unified way.

The quantifiers~$\dexists$ and~$\dforall$ are also binders, and to deal with these binders,
we used standard unscoped de Bruijn syntax. Notice that it is much easier
to implement de Bruijn for first-order logic than it is for higher-order logics or even just lambda calculus,
the reason being that (first-order) terms used for substitutions do not contain binders.
It could have been possible to use the Autosubst 2~\cite{AutoSubst2} support for de Bruijn syntax,
but we refrained from this choice because of its current (technical but not fundamental) dependency
on the axiom of functional extensionality. We remind the reader that in the context of synthetic
undecidability, axioms cannot be freely added without risking breaking the requisite of
computability of the terms constructed in the underlying type theory.

So far, we discussed technical implementation choices that have little or no impact on the
meaning or provability of the reduction results that we implement in this paper. Other choices
are clearly possible, but some could lead to considerably more complicated proofs, e.g.\
setoid hell if dependent types are managed too naively.

\smallskip

At the other end are choices that
could/would impact the provability of our reduction results, and paramount to them, the notions
involved in the definition of what is a finite model of first-order logic. In classical
settings, there is usually no discussion on how to interpret function symbols:
as set-theoretic functions in the model, i.e.\ total and functional binary relations.
However in a constructive setting, already the notion of finiteness has several
non-equivalent implementations. We choose to define finiteness by basing on
the inductive type of lists, used to enumerate members of finite types.
Hence finiteness of $D$ just requires that the terms of type $D$ can
be collected into a list, i.e.\ there is $l:\List D$ such that $\forall x:D.\,x\in l$.

We could have required the domain $D$ of the model to also be discrete (i.e.\ equipped with a computable
way to discriminate elements which are not identical), but we devoted Section~\ref{sec:discrete_models}
to establish that this requirement is not necessary and does not impact the reduction results
(see Theorem~\ref{thm:quotient_FSAT}).
While finiteness  implies discreteness in a purely predicative setting like
the one of Agda where ${\in}: D\to\List D\to\Type$, it does not imply discreteness
when the membership predicate ${\in}: D\to\List D\to\Prop$ is typed in the
impredicative sort $\Prop$.\footnote{because the membership predicate typed ${\in}: D\to\List D\to\Prop$
cannot be used to computationally recover the position of an element in a list, whereas
it can when typed as ${\in}: D\to\List D\to\Type$.}

This illustrates that finiteness alone does not imply computability in the
impredicative constructive setting of Coq. Even on the finite type $\Fin n=\{0,\ldots,n-1\}$,
relations $\Fin n\to\Prop$ are not necessarily decidable. We choose to assume that the
intended meaning of a finite model of first-order logic is that it can be described
\emph{in extension} by tables. Hence, we found appropriate to interpret a function symbol
$f$ by a term of type $D^{|f|}\to D$, that if applied to the list enumerating $D$ and
then reduced on each instance, would yield the intended table.
For relation symbols $P$, we used a term $P^\MM$
of type $D^{|f|}\to \Prop$, \emph{moreover assuming that $P^\MM$ is decidable} (or equivalently
Boolean). This ensures that not only functions but also relations can be described by tables.

However, dropping the decidability of $P^\MM$, or else interpreting function symbols
by functional and weakly total\footnote{where $R:X\to Y\to\Prop$ is \emph{weakly total} if
$\forall x\exists y\, R\,x\,y$ holds.}
relations of type $D^{|f|}\to D\to\Prop$, would certainly
break some of our reductions, again because of the impredicativity of sort $\Prop$.
The question of whether this one relaxation or other relaxations would still be
acceptable renderings of the notion of finite model in constructive type theory
could be debated informally, but formally they could change the meaning of
Trakhtenbrot's theorem up to a breaking point.

\subsection{Future work}

We refrained from additional axioms since we included our development in the growing \href{https://github.com/uds-psl/coq-library-undecidability}{Coq library of synthetic
undecidability proofs}~\cite{library_coqpl}.
In this context, we plan to generalise some of the intermediate signature reductions so that they become reusable for
other undecidability proofs concerning first-order logic over arbitrary models.
For these general reductions however, additional assumptions like unique choice for the encoding of functions as relations or the requirements for a model of set theory as described in~\cite{kirst2019categoricity} to compress to the binary signature will play a role.

As further future directions, we want to explore whether a more direct reduction into the binary signature can be given, circumventing the intermediate signature transformations.
Also possible, though rather ambitious, would be to mechanise the classification of first-order satisfiability with regards to the quantifier prefix as
comprehensively developed in~\cite{borger1997classical}.


\bibliographystyle{alphaurl}
\bibliography{paper}

\appendix

\end{document}

%% file: macros.tex
\makeatletter

\@ifundefined{clap}\newcommand\renewcommand{\clap}[1]{\hbox to 0pt{\hss{#1}\hss}}

\makeatother

\newcommand{\MBB}[1]{\ensuremath{\mathbb{#1}}\xspace}
\newcommand{\MCL}[1]{\ensuremath{\mathcal{#1}}\xspace}
\newcommand{\MSF}[1]{\ensuremath{\mathsf{#1}}\xspace}

\newcommand{\Prop}{\MBB{P}}
\newcommand{\Type}{\MBB{T}}

\newcommand{\cdef}{\mathbin{:=}}
\newcommand{\bnfdef}{\mathbin{::=}}

\newcommand{\toot}{\mathrel\leftrightarrow}

\DeclareMathAlphabet{\mymathbb}{U}{bbold}{m}{n}
\newcommand{\Void}{\mymathbb{0}}
\newcommand{\Unit}{\mymathbb{1}}
\newcommand{\Nat}{\MBB{N}}
\newcommand{\Val}{\MBB{V}}
\newcommand{\Bool}{\MBB{B}}

\newcommand{\SigType}[2]{\Sigma{#1}.\,{#2}}
\renewcommand{\SigType}[2]{\{{#1}\mid{#2}\}}
\newcommand{\Fin}[1]{\MBB{F}_{#1}}

\newcommand{\Term}{\MSF{Term}}
\newcommand{\Formula}{\MSF{Form}}

\newcommand{\dbot}{\dot\bot}
\newcommand{\dtop}{\dot\top}
\newcommand{\dand}{\dot\land}
\newcommand{\dor}{\dot\lor}
\newcommand{\dto}{\dot\to}
\newcommand{\dtoot}{\dot\toot}
\newcommand{\dneg}{\dot\neg}
\newcommand{\dforall}{\dot\forall}
\newcommand{\dexists}{\dot\exists}
\newcommand{\din}{\mathrel{\dot\in}}

\DeclareMathOperator*{\dbigvee}{\scalerel*{\bigvee}{\sum}}

\newcommand{\arity}[1]{\mathalpha{|{#1}|}}

\newcommand{\unit}{\mathtt{*}}
\newcommand{\btrue}{\mathsf{tt}}
\newcommand{\bfalse}{\mathsf{ff}}

\newcommand{\some}[1]{\ulcorner#1\urcorner}
\newcommand{\none}{\emptyset}

\newcommand{\Opt}{\MBB{O}}
\newcommand{\List}{\MBB{L}}

\newcommand{\capp}{\mathbin{+\hspace{-6pt}+}} 
\newcommand{\map}{\,}
\newcommand{\inl}{\in}
\newcommand{\incl}{\subseteq}
\newcommand{\clen}[1]{\mathalpha{|{#1}|}}
\newcommand{\cnil}{\mathalpha{[\,]}}
\newcommand{\ccons}{\mathbin{::}}

\newcommand{\subst}[3]{{#1}[{#2}/{#3}]}

\newcommand{\tapp}{\mathbin{+\hspace{-8pt}+\hspace{-8pt}+}}

\newcommand{\red}{\mathrel\preceq}
\newcommand{\BPCP}{\MSF{BPCP}}

\newcommand{\SBPCP}{\Sigma_{\BPCP}}
\newcommand{\deriv}{\triangleright}

\newcommand{\funcssymb}{\MCL{F}}
\newcommand{\predssymb}{\MCL{P}}

\newcommand{\funcs}[1]{\funcssymb_{#1}}
\newcommand{\preds}[1]{\predssymb_{#1}}

\newcommand{\Funcs}{\funcs\Sigma}
\newcommand{\Preds}{\preds\Sigma}
\newcommand{\FV}{\mathsf{FV}}

\renewcommand{\phi}{\varphi}
\newcommand{\binop}{\mathbin{\dot\square}}
\newcommand{\binopm}{\Box}
\newcommand{\quant}{\dot\nabla}
\newcommand{\quantm}{\nabla}

\newcommand{\MM}{\MCL{M}}
\newcommand{\BB}{\MCL{B}}

\newcommand{\SAT}{\MSF{SAT}}
\newcommand{\FSAT}{\MSF{FSAT}}
\newcommand{\FSATEQ}{\MSF{FSATEQ}}

\newcommand{\MSLSAT}{\MSF{MSLSAT}}
\newcommand{\SLSAT}{\MSF{SLSAT}}
\newcommand{\MSL}{\MSF{MSL}}
\newcommand{\SL}{\MSF{SL}}
\newcommand{\emp}{\MSF{emp}}

\newcommand{\sep}{-\kern-.6em\raisebox{-.659ex}{*}\ }